\documentclass[12pt,twocolumn,twoside]{IEEEtran}
\usepackage{cite}
\usepackage{graphicx}
\usepackage{epstopdf}
\usepackage[cmex10]{amsmath}
\usepackage{stfloats}
\usepackage{subfigure}
\usepackage{multirow}
\usepackage{url}
\usepackage{amssymb}
\usepackage{amsmath}
\usepackage{balance}
\usepackage{diagbox}
\usepackage{lineno}
\usepackage{float}
\usepackage{gensymb}
\usepackage{xcolor}

\usepackage{amsthm}
\theoremstyle{plain}
\newtheorem{thm}{\protect\theoremname}
\theoremstyle{remark}
\newtheorem{rem}{\protect\remarkname}
\theoremstyle{plain}
\newtheorem{prop}{\protect\propositionname}
\theoremstyle{plain}
\newtheorem{lem}{\protect\lemmaname}
\usepackage{algorithm}
\usepackage{algpseudocode}
\usepackage{algorithmicx}
\usepackage{textcomp,flushend}
\usepackage{babel}
\providecommand{\lemmaname}{Lemma}
\providecommand{\propositionname}{Proposition}
\providecommand{\remarkname}{Remark}
\providecommand{\theoremname}{Theorem}

\begin{document}
\title{Learning Rate Optimization for Federated Learning Exploiting Over-the-air Computation}
\author{Chunmei Xu,~\IEEEmembership{Student~Member,~IEEE},
        Shengheng~Liu,~\IEEEmembership{Member,~IEEE},
        Zhaohui~Yang,~\IEEEmembership{Member,~IEEE},
        Yongming~Huang,~\IEEEmembership{Senior Member,~IEEE},
       % Dusit~Niyato,~\IEEEmembership{Fellow,~IEEE},
        Kai-Kit~Wong,~\IEEEmembership{Fellow,~IEEE}
        \vspace{-1.2em}

%\thanks{Manuscript received XXX XX, 2021; revised XXX XX, XXXX; accepted XXX XX, XXXX. Date of publication XXX XX, XXXX; date of current version XXX XX, XXXX. This work was supported in part by the National Natural Science Foundation of China under Grant Nos. 62001103 and 61720106003, the National Key R\&D Program of China under Grant No. 2018YFB1800801. (Corresponding authors: S.~Liu and Y.~Huang.)}

\thanks{C.~Xu, S.~Liu, and Y.~Huang are with the School of Information Science and Engineering, Southeast University, Nanjing 210096, China, and also with the Purple Mountain Laboratories, Nanjing 211111, China (e-mail: \{xuchunmei; s.liu; huangym\}@seu.edu.cn).}

\thanks{Z.~Yang is with the Centre for Telecommunications Research, Department of Engineering, King's College London, WC2R 2LS, UK, (e-mail: yang.zhaohui@kcl.ac.uk).}

%\thanks{D.~Niyato is with the School of Computer Science and Engineering, Nanyang Technological University, Singapore 639798 (e-mail: dniyato@ntu.edu.sg).}

\thanks{K.-K.~Wong is with the Department of Electronic and Electrical Engineering, University College London, London WC1E 6BT, United Kingdom (e-mail: kai-kit.wong@ucl.ac.uk).}
}

\maketitle
\begin{abstract}
Federated learning (FL) as a promising edge-learning framework can effectively address the latency and privacy issues by featuring distributed learning at the devices and model aggregation in the central server. {In order to enable efficient wireless data aggregation, over-the-air computation (AirComp) has recently been proposed and attracted immediate attention. However, fading of wireless channels can produce aggregate distortions in an AirComp-based FL scheme. To combat this effect, the concept of dynamic learning rate (DLR) is proposed in this work.} We begin our discussion by considering multiple-input-single-output (MISO) scenario, since the underlying optimization problem is convex and has closed-form solution. We then extend our studies to more general multiple-input-multiple-output (MIMO) case and an iterative method is derived. Extensive simulation results demonstrate the effectiveness of the proposed scheme in reducing the aggregate distortion and guaranteeing the testing accuracy using the MNIST and CIFAR10 datasets. In addition, we present the asymptotic analysis and give a {near-optimal} receive beamforming design solution in closed form, which is verified by numerical simulations.
\end{abstract}

\begin{IEEEkeywords}
Distributed algorithm, federated learning, over-the-air computation, learning rate, beamforming.
\end{IEEEkeywords}

\section{Introduction}

Future sixth-generation (6G) {communication networks are} envisioned to undergo {a profound} transformation, which {evolves from} \emph{connected
things} to \emph{connected intelligence} with more stringent requirements {such as dense networking, strict security, high energy efficiency, and high intelligence \cite{Roadmap, Huang21}. Artificial intelligence (AI) technologies, which allows automatic analysis of a large {mass} of data generated in wireless networks and subsequent optimization of highly dynamic and complex network \cite{6Gvision, deep_survey, Xu21}, will shape the landscape of 6G. Conversely, 6G will give renewed impetus to the AI-empowered mobile applications by supplying the} advanced wireless communications and mobile computing technologies \cite{what6gbe} {as supporting infrastructure}.

{AI tasks entail increasingly intensive computation workloads. Hence, they are generally migrated to and trained on the server center with sufficient computation resources and} the availability of data {that} is first collected from the devices/sensors and then {uploaded to the center} \cite{Cloud_ML1,Cloud_ML2}. The data volume can be considerably
large and, thus, imposing heavy transmission traffic burden and increasing
the latency. {Another critical problem comes from the serious concern of privacy leakage,} since the generated data, e.g., photos,
social-networking records, at the devices are often privacy sensitive.
An intuitive way to counteract the above issues would be to conduct
training and inference process directly at the network edge, such
as devices and sensors, using locally generated real-time data. The
{paradigm} edge learning has unique advantages
of balanced resource support, proximity to data sources compared with
cloud learning and higher learning accuracy compared {to} on-device
learning by harnessing the computation and storage capacities \cite{edge,edge3}.

{Federated learning (FL) tackles the aforementioned concerns by} collaboratively
training a shared global model with locally stored data \cite{2016Communication,FL_6G, FL_6G_2, FL_IoT}.
A typical FL algorithm alternates between two iterative phases: (I) The devices
receive the global model from the edge server and train local models
with locally stored data; (II) These local models are transmitted
to and aggregated at the edge server to yield the global model. Note
that the data volume of the local models (may consist of millions
of parameters) are much smaller than {the raw data.} {Nonetheless}, the
local {models} uploaded by {legion of} participating devices via wireless
links is resource-demanding, which {is} the main bottleneck
to implement FL in practice. In this regard, developing communication-efficient
methods are of paramount importance. Some recent works have considered
asynchronous mechanism \cite{Parallelism}, quantization \cite{quantization,quantization1},
sparsification \cite{compression,compression1}, and aggregate frequency
\cite{Nishio} to reduce the transmission overhead, which, however,
ignore the aspects of physical and network layers.

{In the second phase of FL, the edge server averages} the local model parameters from {the} distributed devices, which {is essentially} wireless data aggregation. Conventional {multiple-access} schemes, such as orthogonal frequency-division multiple access
(OFDM), are based on the separated-communication-and-computing principle.
In \cite{TDMA}, a time-division multiple access (TDMA) system was
considered, where a joint batchsize selection and communication resource
allocation scheme was developed aiming at accelerating the training
process and improving the learning efficiency. The impact of three
different scheduling policies on the FL performance were analytically
studied in the large-scale wireless networks \cite{scheduling}. Such
sub-optimal communication-and-computation approaches could result
in a sharp rise in consumption of wireless resources as well as congesting
the air-interface \cite{edge}. {Very Recently, the over-the-air computation (AirComp) scheme was proposed by leveraging the waveform superposition property of multiple-access channels, which is fundamentally different from the traditional separated-communication-and-computing principle \cite{zhu2020overtheair}.} By aggregating the data simultaneously received from distributed devices in an analog manner, the AirComp technique can further improve communication efficiency
\cite{Gunduz, Gunduz2, Broadband}.

The AirComp technique has recently been applied to implementing FL. Specifically, a gradient sparsification and random linear projection based approach is proposed and the reduced data was transmitted via AirComp to address the bandwidth and power limitations, which outperformed its digital
counterpart \cite{Gunduz,Gunduz2}. As a matter of fact, the distortions caused
by fading and noisy channels are critical for learning tasks as a
large aggregation error may lead to the degradation of inference accuracy.
In \cite{Broadband}, the transmission power was designed by truncated
channel inversion, and two scheduling schemes were proposed to guarantee
the identical amplitudes of the received signals among devices in
a single-input-single-output (SISO) system to reduce the aggregate
error. {Meanwhile in \cite{kaiyang}, joint device selection and receive beamforming design was investigated in single-input-multiple-output (SIMO) configuration, and} a novel unified difference-of-convex {(DC)} function was proposed. {On the other hand, the problem of distortion minimization in an intelligent reflection surface (IRS)-aided cloud radio access network (CRAN) system was addressed in \cite{IRS}, where a joint optimization scheme of the passive beamforming and linear detection vector was designed. Furthermore, with} the aid of multi-IRS, a novel resource and device selection method was developed to minimize the aggregate error as
well as maximize the selected devices \cite{WanliNi}. The existing
works utilized wireless resources, such as power control, device selection
and beamforming design, as well as channel configuration to align
the received signals from distributed devices. Nevertheless, they did not fully unleash
the potential of hyper-parameters in the perspective of machine learning
(ML).

Learning rate is a key hyper-parameter that determines the convergence
and the convergence rate of the learning tasks. A large learning rate will hinder convergence and cause loss function around the minimum or even to diverge, while
too small a learning rate will lead to slow convergence \cite{lr_trick}. To select an optimal learning rate is always critical, yet tricky issue for learning algorithms to work properly. One feasible approach is to adopt learning rate schedulers, which is able to adjust the learning rate
training online. However, it has to be designed in advance and is unable
to adapt to the characteristics of the dataset \cite{scheduler_lr}.
{Adaptive learning rates such as Adagrad can adapt to the data and change with the gradients, which are most widely used in deep learning community \cite{adapt_lr}.} {Later, cyclical} learning rate (CLR) was {proposed} to allow the learning rate cyclically vary between reasonable boundary values, which {incurs less computational cost and can} significantly enhance the learning performance \cite{cyclical_lr}. The essence behind CLR originates from the observation that increasing
the learning rate allows more rapid traversal of saddle point plateaus
and thus achieves a longer term beneficial effect. Inspired by this study,
we propose dynamic learning rate (DLR) between the minimum and maximum
boundaries to adapt to the fading channels, in order to further reduce
the aggregate error caused by the fading and noisy channels.

In this paper, we consider FL for AI-empowered mobile applications,
such as e-health services, which will be supported by 6G networks.
Instead of adopting conventional separated-communication-and-computation
pattern, we incorporate AirComp to aggregate local models from distributed
devices so as to improve the communication efficiency. In AirComp-based schemes, minimizing the resultant aggregate distortion is of paramount
importance as a large distortion can {spawn} performance degradation
of AI tasks. To mitigate the wireless distortion measured by mean
square error (MSE), we first propose to utilize a DLR scheme to adapt
to the wireless channels, and receive beamforming
optimization is jointly considered.  The technical contributions of this work are summarized below.
\begin{itemize}
\item To our best knowledge, this is the first work to study
the use of DLR for FL over wireless communications to reduce the
aggregate error, which is fundamentally different from existing
works which only consider the optimization of wireless resources. We
analytically show that DLR can be properly designed to mitigate the
distortion caused by fading. It is also proved that the MSE can be
further decreased by considering DLR.
\item For MSE minimization via AirComp, we jointly optimize the DLR ratios
and wireless resources. Both MISO and MIMO scenarios are considered, and the respective closed-form solution and iterative algorithm are developed. Extensive simulation results demonstrate
the effectiveness of the proposed scheme in further reducing the MSE as well as {improving} the learning performance on MNIST and CIFAR10 datasets.
\item In addition, we present the asymptotic beamforming solution in closed form when the number of transmit/receive antennas tends to infinity. Simulation results verify the theoretical analysis as well as the receive beamforming design.
\end{itemize}

The outline of this paper is organized as follows: Some necessary mathematical descriptions of FL and AirComp are presented in Section~\ref{S:pre}. The concept of DLR is introduced in Section~\ref{S:DLR}. In Section~\ref{S:ProFor} and Section~\ref{S:DLRsov}, the DLR optimization problems in MISO and MIMO scenarios are respectively formulated and solved. Next, we present the theoretical asymptotic analysis and, on this basis, propose a {near-optimal} and closed-form receive beamforming solution in Section~\ref{S:AAna}.
Then, in Section~\ref{S:SR}, numerical simulations are provided to showcase the advantages of the proposed scheme. Finally, the paper is concluded in Section~\ref{S:conclu}.

\section{Preliminary}\label{S:pre}

{In this work, we consider the problem of FL over wireless networks. The configuration of the system under investigation is depicted in Fig.~\ref{fig:System-model}. The wireless network consist} of $K$ devices with $N_{d}$ antennas each and an aggregator with $N_{t}$ antennas. {The set of devices are denoted as $\mathcal{K}$. }Each device
$k$ updates its model based on locally distributed data $\mathcal{D}_{k}$, which cannot be shared with other entities out of latency and privacy
concerns.

\begin{figure}[!htpb]
\centering
\includegraphics[width=0.99\columnwidth]{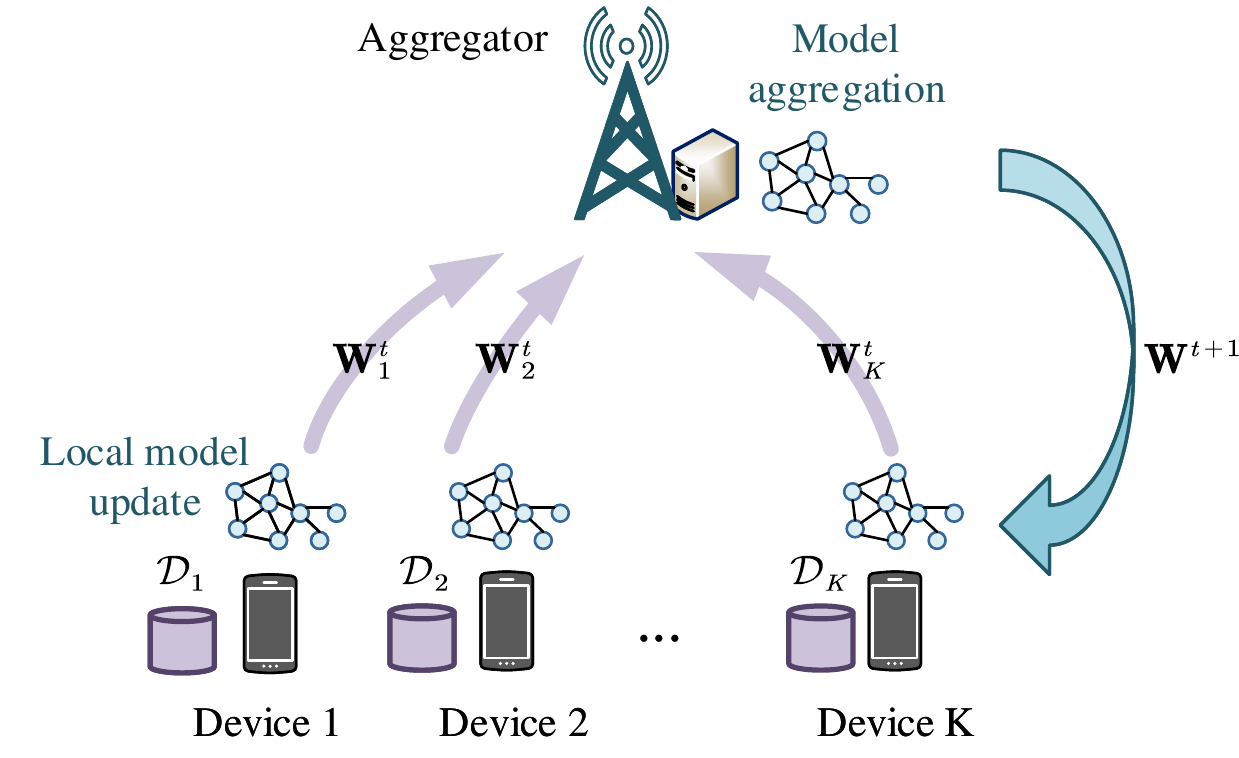}
\caption{An FL system over wireless communication.}
\label{fig:System-model}
\end{figure}

\subsection{FL}

FL has recently emerged as an effective distributed approach to enable
wireless devices to collaboratively build a shared learning model
with training taken place locally. The objective of FL is to minimize
the aggregate loss:
\begin{equation}
\boldsymbol{w}^{o}\triangleq\underset{\boldsymbol{w}}{\mathrm{argmin}}\frac{1}{K}\sum\nolimits_{k=1}^{K}P_{k}\left(\boldsymbol{w}\right),\label{eq:centralized}
\end{equation}
where $\boldsymbol{w}\in\mathbb{R}^{D}$ is a vector containing the
model parameters, $D$ is the dimension of the FL model, $P_{k}\left(\boldsymbol{w}\right)$
is the local loss value at device $k$ based on $\mathcal{D}_{k}$, given by
\begin{equation}
P_{k}\left(\boldsymbol{w}\right)\triangleq\frac{1}{\left|\mathcal{D}_{k}\right|}\sum\nolimits_{n=1}^{\left|\mathcal{D}_{k}\right|}Q_{k}\left(\boldsymbol{w};\boldsymbol{x}_{n},y_{n}\right),
\end{equation}
where $Q_{k}$ is the loss function on the sample $\left(\boldsymbol{x}_{n},y_{n}\right)$ with $\boldsymbol{x}_{n},y_{n}$ the input and the label. To obtain the solution of (\ref{eq:centralized}),
a centralized gradient decent method is applied and the parameters
are updated as
\begin{equation}
\boldsymbol{w}^{i+1}=\boldsymbol{w}^{i}-\mu\left(\frac{1}{K}\sum\nolimits_{k=1}^{K}\boldsymbol{g}_{k}\left(\boldsymbol{w}^{i}\right)\right),\label{eq:update}
\end{equation}
where $i$ is the iteration index, $\mu$ is the learning rate, $\boldsymbol{g}_{k}(\boldsymbol{w}^{i})=\nabla_{\boldsymbol{w}}P_{k}\left(\boldsymbol{w}^{i}\right)$
is the gradient of loss with respect to $\boldsymbol{w}^{i}$. Hereinafter,
we denote $\boldsymbol{g}_{k}(\boldsymbol{w}^{i})$ by $\boldsymbol{g}_{k}$ for notation simplicity.
Since the aggregator is inaccessible to the data distributed at any
particular device $k$, the gradient term $\boldsymbol{g}_{k}$ is
calculated locally and the local model, denoted as $\boldsymbol{w}_{k}$,
is updated accordingly. By introducing the local learning rate denoted
as $\mu_{k}$, the local model at device $k$ is updated by
\begin{equation}
\boldsymbol{w}_{k}^{i+1}=\boldsymbol{w}^{i}-\mu_{k}\boldsymbol{g}_{k}.
\end{equation}
With local models received at the aggregator, the update of the global
model (\ref{eq:update}) is rewritten as
\begin{equation}
\boldsymbol{w}^{i+1}=\frac{1}{K}\sum\nolimits_{k=1}^{K}\boldsymbol{w}_{k}^{i+1}.\label{eq:update_1}
\end{equation}

\subsection{AirComp}
{As introduced earlier, AirComp integrates computation and communication by exploiting the waveform superposition property, which harnesses interference to help functional computation \cite{zhu2020overtheair}.}
The AirComp {comprises three stages}: (i) Pre-processing at the
transmitters; (ii) Superposition over the air; and (iii) Post-processing
at the receiver \cite{Overtheair}. In this work, the pre-processing
is assumed to be identity mapping. Each parameter in a local model is modulated as a symbol, then the symbol vector $\boldsymbol{s}_{k}^{i+1}\triangleq\boldsymbol{w}_{k}^{i+1}\in\mathbb{C}^{D}$ is obtained accordingly. The symbol vector is assumed to be normalized {by the} unit variance, {which is given by} $\mathbb{E}\left[\boldsymbol{s}_{k}^{i+1}\left(\boldsymbol{s}_{k}^{i+1}\right)^{\mathrm{H}}\right]=\mathbf{I}$. For notation simplicity,
the $d$-th element of $\boldsymbol{s}_{k}$, $\boldsymbol{w}^{i}$
and $\boldsymbol{g}_{k}$, i.e., $\boldsymbol{s}_{k}[d]$, $\boldsymbol{w}^{i}\left[d\right]$
and $\boldsymbol{g}_{k}[d]$, {are written as $s_{k}$,
$w^{i}$ and $g_{k}$}. As such, the desired signal based on (\ref{eq:update})
can be represented {by}
\begin{equation}
y_{\mathrm{des}}\!=\!\frac{1}{K}\!\sum\nolimits_{k=1}^{K}\!s_{k}\!=\!\frac{1}{K}\!\sum\nolimits_{k=1}^{K}\!w_{k}^{i+1}\!=\!w^{i}\!-\frac{\mu}{K}\sum\nolimits_{k=1}^{K}\!g_{k}.
\end{equation}
Considering the multiple access channel property of wireless communication,
the received signal is a linear sum of {the} transmitted signal plus uncertainty.
{Hence, after the} post-processing, the received signal at the receiver {can be expressed as}
\begin{equation}
y=\sqrt{\eta}\left(\sum\nolimits_{k=1}^{K}A_{k}s_{k}+B\right),\label{eq:aggregat_sig}
\end{equation}
where $s_{k}$ is the input of the communication system from device
$k${,} and $\eta$ is a scaling factor. The {variables} $A_{k}$ and \textbf{$B$}
depend on the {specific} scenario {settings. In particular,} we have $A_{k}=h_{k}b_{k}$ {and} $B=n$ for {SISO;} $A_{k}=\boldsymbol{h}_{k}^{\mathrm{T}}\boldsymbol{b}_{k}$ {and} $B=n$ for MISO. For SIMO and MIMO scenarios, we have $A_{k}=\boldsymbol{m}^{\mathrm{H}}\boldsymbol{h}_{k}b_{k}$,
$B=\boldsymbol{m}^{\mathrm{H}}\boldsymbol{n}$, and $A_{k}=\boldsymbol{m}^{\mathrm{H}}\boldsymbol{H}_{k}\boldsymbol{b}_{k}$,
$B=\boldsymbol{m}^{\mathrm{H}}\boldsymbol{n}$, respectively. Note
that $h_{k}\left(\boldsymbol{h}_{k},\boldsymbol{H}_{k}\right)$ is
the independent Rayleigh fading channel (vector/matrix) from device
$k$ to the aggregator, which is assumed to be block fading and {remain constant during the process of model} transmission; $b_{k}\left(\boldsymbol{b}_{k}\right)$ is
the transmit coefficient (vector) at device $k$; $\boldsymbol{m}$
is the receive beamforming vector at the aggregator; $n\left(\boldsymbol{n}\right)$
is the noise (vector) with power of $\sigma^{2}$.

Instead of obtaining the individual term $s_{k}$ first and then {averaging} at the aggregator, we {employ AirComp to estimate} the desired signal directly. The {AirComp aggregate error is defined as} the difference between the {desired signal and its estimation}, which
{is derived} as
\begin{equation}
\begin{array}{*{20}{l}}
e\triangleq \displaystyle y_{\mathrm{des}}-y=\frac{1}{K}\sum\nolimits_{k=1}^{K}s_{k}-\sqrt{\eta}\left(\sum\nolimits_{k=1}^{K}A_{k}s_{k}+B\right)\\
\;\;\,\displaystyle=\sum\nolimits_{k=1}^{K}\left(\frac{1}{K}-\sqrt{\eta}A_{k}\right)s_{k}-\sqrt{\eta}B.
\end{array}
\label{eq:error}
\end{equation}
According to (\ref{eq:error}), the aggregate distortion originates
from both fading and noise, which are reflected in the {fading-channel-related} term $\sqrt{\eta}A_{k}$ and the noise-related term \textbf{$\sqrt{\eta}B$}{,}
respectively. It is worth noting that mitigating the
distortion is of great importance, as {severe distortion can result in a} biased global model and, in turn, degrades the learning performance.

\section{DLR for Channel Adaption}\label{S:DLR}

{Existing works \cite{IRS,kaiyang,WanliNi,Broadband} minimize the aggregate error by means of wireless resource optimization, IRS-based channel reconfiguration, and device selection. These approaches correspond to optimize the transmit coefficient
(vector) $b_{k}\left(\boldsymbol{b}_{k}\right)$, the receive beamforming
vector $\boldsymbol{m}$, the scaling factor $\eta$, and the passive
beamforming vector $\Theta$, or simply select a subset of the devices. Though appears different at first glance, these approaches share a common aim, which is to align the receive signals and minimize the noise-induced error. Whereas in this work, we take a radically different perspective and propose to mitigate the distortions by optimizing the hyper-parameters in the learning process. Concretely, a strategy is designed to let the local learning rate $\text{\ensuremath{\mu_{k}}}$ adapt to the time-varying wireless environment.} Based on the aforementioned features of FL and AirComp,
we arrive at the following theorem.
\begin{thm}
\label{thm:1} Denote DLR ratio as $r_{k}=\frac{\mu_{k}}{\mu}$
for device $k$, to mitigate the distortion caused by the fading channel related term $\sqrt{\eta}A_{k}$, we have
\begin{equation}
\sqrt{\eta}\sum\nolimits_{k=1}^{K}A_{k}=1,\quad r_{k}=\frac{1}{K\sqrt{\eta}A_{k}}.\label{eq:ther}
\end{equation}
\end{thm}
\begin{proof}
The aggregate error $e^{ch}$ caused by the fading channels can
be written as
\begin{align}
e^{ch} & =y_{\mathrm{des}}-\sqrt{\eta}\sum\nolimits_{k=1}^{K}A_{k}s_{k}\nonumber \\
&=\left(w^{i}\!-\!\frac{\mu}{K}\sum\nolimits_{k=1}^{K}g_{k}\right)\!-\!\sqrt{\eta}\sum\nolimits_{k=1}^{K}A_{k}\left(w^{i}\!-\!\mu_{k}g_{k}\right)\nonumber \\
 & =\left(1\!-\!\sqrt{\eta}\sum_{k=1}^{K}A_{k}\right)w^{i}\!+\!\sum_{k=1}^{K}\left(\sqrt{\eta}A_{k}\mu_{k}\!-\!\frac{\mu}{K}\right)g_{k},
\end{align}
which is mitigated if and only if both terms $\left(\frac{1}{K}\mu-\sqrt{\eta}A_{k}\mu_{k}\right)$
and $\left(1-\sqrt{\eta}\sum_{k=1}^{K}A_{k}\right)$ are $0$. Thus,
we directly obtain (\ref{eq:ther}).
\end{proof}
Based on Theorem \ref{thm:1}, the residual aggregate error is the
noise-related term $\sqrt{\eta}B$, which can be {measured by}
\begin{equation}
\mathrm{MSE}=\eta\mathbb{E}\left(\left\Vert B\right\Vert ^{2}\right).\label{eq:MSE}
\end{equation}
  For simplicity, we consider the retransmission mechanism {such} that
if there exists aggregate error, {the} probability of retransmission
{is}
\begin{equation}
P=1-\exp\left(-\frac{a\mathrm{\left\Vert \mathit{e}\right\Vert ^{2}}}{p_{\mathrm{des}}}\right),\label{eq:PER}
\end{equation}
where $a$ is the {modulation-related parameter} \cite{chen2020joint},
$p_{\mathrm{des}}$ {denotes} the power of the desired signal, {and} $e$ is the
aggregate error. Intuitively, a larger aggregate error {leads}
to a larger retransmission rate.

\section{Problem Formulation}\label{S:ProFor}

{The objective is to minimize the MSE metric given in (\ref{eq:MSE}), subject to equality constraint (\ref{eq:ther}) and boundary constraint $r_{k}=\mu_{k}/\mu\in\left[r_{\min},r_{\max}\right]$. In this section, we establish the problem formulations for both MISO and MIMO cases, which will be shown in the sequel are respectively convex and nonconvex. It is important to note} that SISO and SIMO can be regarded as the special cases of MISO and MIMO scenarios, respectively.

\subsection{MISO}

In the MISO scenario, the devices equipped with $N_{d}$ antennas
each transmit their models to the single-antenna aggregator. Under
this scenario, we have $A_{k}=\boldsymbol{h}_{k}^{\mathrm{T}}\boldsymbol{b}_{k}$ {and} $B=n$.
The aggregate error measured by MSE is then given by
\begin{equation}
\mathrm{MSE}=\eta\mathbb{E}\left(\left\Vert B\right\Vert ^{2}\right)=\eta\mathbb{E}\left(\left\Vert n\right\Vert ^{2}\right)=\eta\sigma^{2}.
\end{equation}
Since the noise power $\sigma^{2}$ is independent from the optimized
variables, the problem can be formulated as
\begin{subequations}\label{eq:SISO1}
\begin{align}
\min_{\eta,b_{k},r_{k}}\quad & \eta\nonumber \\
\mathrm{s.t.}\quad & \sqrt{\eta}\sum\nolimits_{k=1}^{K}\boldsymbol{h}_{k}\boldsymbol{b}_{k}=1, \label{eq:eqc1}\\
 & r_{k}=\frac{1}{K\sqrt{\eta}\boldsymbol{h}_{k}\boldsymbol{b}_{k}},  \quad\forall k,\label{eq:eqc}\\
 & r_{k}\in\left[r_{\min},r_{\max}\right], \quad\forall k,\\
 & \left\Vert \boldsymbol{b}_{k}\right\Vert ^{2}\leq P_{k},  \quad\forall k,\label{eq:powc}
\end{align}
\end{subequations}
where $P_{k}$ is the maximum transmit power at device $k$, {and} $\boldsymbol{h}_{k}\in\mathbb{C}^{N_{d}}$
is the channel vector between device $k$ and the aggregator. Both
equality constraints (\ref{eq:eqc1}) and (\ref{eq:eqc}) conspire
to guarantee the elimination of error $e^{\mathrm{ch}}$ caused by
the wireless fading channels as per \textbf{Theorem} \ref{thm:1}.

Motivated by the uniform-forcing transceiver design in \cite{8364613},
the optimal transmitting coefficient vector $\boldsymbol{b}_{k}$
is designed as
\begin{equation}
\boldsymbol{b}_{k}=\frac{\boldsymbol{h}_{k}^{\mathrm{H}}}{K\sqrt{\eta}\left\Vert \boldsymbol{h}_{k}\right\Vert ^{2}r_{k}}.
\end{equation}
Power constraint (\ref{eq:powc}) further suggests that $\frac{1}{K^{2}P_{k}r_{k}^{2}\left\Vert \boldsymbol{h}_{k}\right\Vert ^{2}}\leq\eta$, and thus {we have}
\begin{equation}
\eta=\max_{k}\frac{1}{K^{2}P_{k}r_{k}^{2}\left\Vert \boldsymbol{h}_{k}\right\Vert ^{2}}.
\end{equation}
{We learn} from (\ref{eq:eqc}) that $\boldsymbol{h}_{k}^{\mathrm{T}}\boldsymbol{b}_{k}=\frac{1}{K\sqrt{\eta}r_{k}}$.
By substituting it back to (\ref{eq:eqc1}), we have $\sum_{k=1}^{K}\frac{1}{Kr_{k}}=1$.
As a result, problem (\ref{eq:SISO1}) is equivalent to
\begin{subequations}\label{eq:SISO2}
\begin{align}
\underset{r_{k}}{\min}\quad & \max_{k}\frac{1}{K^{2}P_{k}r_{k}^{2}\left\Vert \boldsymbol{h}_{k}\right\Vert ^{2}}\\
\mathrm{s.t.}\quad & \sum\nolimits_{k=1}^{K}\frac{1}{Kr_{k}}=1, \\
 & r_{k}\in\left[r_{\min},r_{\max}\right],  \quad\forall k.\label{eq:MISO_cons1}
\end{align}
\end{subequations}
\begin{rem}
As a special case of the MISO scenario where $A_{k}=h_{k}b_{k}$
and $B=n$, the problem formulated under the SISO case is similar
to (\ref{eq:SISO2}). The only difference lies in the channel and
transmit coefficients, which are both complex scalars instead of vectors
in the MISO case.
\end{rem}

\subsection{MIMO}

In the MIMO scenario, each device and the aggregator are equipped
with $N_{d}$ and $N_{t}$ antennas, respectively. The terms $A_{k}$
and $B$ in (\ref{eq:aggregat_sig}) are $A_{k}=\boldsymbol{m}^{\mathrm{H}}\boldsymbol{H}_{k}\boldsymbol{b}_{k}$,
$B=\boldsymbol{m}^{\mathrm{H}}\boldsymbol{n}$ with $\boldsymbol{m}$
the receive beamforming vector. Accordingly, the aggregate error measured
by MSE is expressed as
\begin{equation}
%{\begin{array}{*{20}{l}}
\mathrm{MSE}\!=\!\eta\mathbb{E}\left(\!\left\Vert B\right\Vert ^{2}\!\right)\!=\!\eta\mathbb{E}\left(\!\left\Vert \boldsymbol{m}^{\mathrm{H}}\boldsymbol{n}\right\Vert ^{2}\!\right)\!=\!\sigma^{2}\!\left\Vert \boldsymbol{m}\right\Vert ^{2}\!\eta,
%\end{array}
\end{equation}
where $\boldsymbol{n}$ the independent Gaussian noise vector. Based
on \textbf{Theorem} \ref{thm:1}, the MSE minimization problem can
be formulated as
\begin{subequations}\label{eq:SIMO}
\begin{align}
\min_{\boldsymbol{m},\eta,r_{k},\boldsymbol{b}_{k}}\quad & \left\Vert \boldsymbol{m}\right\Vert ^{2}\eta\label{eq:pro1_1-1}\\
\mathrm{s.t.\quad} & \sqrt{\eta}\sum\nolimits_{k=1}^{K}\boldsymbol{m}^{\mathrm{H}}\boldsymbol{H}_{k}\boldsymbol{b}_{k}=1,\label{eq:pro1-2-1}\\
 & r_{k}=\frac{1}{K\sqrt{\eta}\boldsymbol{m}^{\mathrm{H}}\boldsymbol{H}_{k}\boldsymbol{b}_{k}} ,  \quad\forall k, \label{eq:pro1_3-1}\\
 & r_{k}\in\left[r_{\min},r_{\max}\right],  \quad\forall k,\\
 & \left\Vert \boldsymbol{b}_{k}\right\Vert ^{2}\leq P_{k},  \quad\forall k,\label{eq:pro1_4-1}
\end{align}
\end{subequations}
{where} $\boldsymbol{b}_{k}\in\mathbb{C}^{N_{d}},\boldsymbol{H}_{k}\in\mathbb{C}^{N_{t}\times N_{d}},\boldsymbol{m\in\mathbb{C}}^{Nt}$
are the transmitting beamforming vector at device $k$, the channel
matrix between the aggregator and device $k$, and the receive beamforming
vector at the aggregator, respectively. According to the constraints (\ref{eq:pro1_3-1}), the optimal transmitting coefficient can be readily {obtained \cite{8364613}}, i.e.,
\begin{equation}
\boldsymbol{b}_{k}=\frac{\boldsymbol{H}_{k}^{\mathrm{H}}\boldsymbol{m}}{K\sqrt{\eta}r_{k}\left\Vert \boldsymbol{m}^{\mathrm{H}}\boldsymbol{H}_{k}\right\Vert ^{2}}.
\end{equation}
Power constraint (\ref{eq:pro1_4-1}) further indicates that $\eta\geq\frac{1}{K^{2}P_{k}r_{k}^{2}\left\Vert \boldsymbol{m}^{\mathrm{H}}\boldsymbol{H}_{k}\right\Vert ^{2}}$
for each device $k$, which implies
\begin{equation}
\eta=\underset{k}{\max}\frac{1}{K^{2}P_{k}r_{k}^{2}\left\Vert \boldsymbol{m}^{\mathrm{H}}\boldsymbol{H}_{k}\right\Vert ^{2}}.
\end{equation}
Similar to the MISO scenario, the ratio $r_{k}$ satisfies $\sum_{k=1}^{K}\frac{1}{Kr_{k}}=1$,
which can be easily derived from the equality constraints (\ref{eq:eqc1})
and (\ref{eq:eqc}). Problem (\ref{eq:SIMO}) can then be rewritten
as
\begin{subequations}\label{eq:SIMO1}
\begin{align}
\min_{\boldsymbol{m},r_{k}}\quad & \underset{k}{\max}\frac{\left\Vert \boldsymbol{m}\right\Vert ^{2}}{K^{2}P_{k}r_{k}^{2}\left\Vert \boldsymbol{m}^{\mathrm{H}}\boldsymbol{H}_{k}\right\Vert ^{2}}\label{eq:pro1_1-1-1}\\
\mathrm{s.t.\quad} & \sum_{k=1}^{K}\frac{1}{Kr_{k}}=1,\label{eq:pro1-2-1-1}\\
 & r_{k}\in\left[r_{\min},r_{\max}\right], \quad\forall k. \label{eq:pro1_3-1-1}
\end{align}
\end{subequations}
\begin{prop}
Problem (\ref{eq:SIMO1}) is equivalent to
\begin{subequations}\label{eq:SIMO2}
\begin{align}
\min_{\boldsymbol{m},r_{k}}\quad & \underset{k}{\max}\frac{\left\Vert \boldsymbol{m}\right\Vert ^{2}}{K^{2}P_{k}r_{k}^{2}\left\Vert \boldsymbol{m}^{\mathrm{H}}\boldsymbol{H}_{k}\right\Vert ^{2}}\label{eq:SIMO2-1}\\
\mathrm{\quad} & (\ref{eq:pro1-2-1-1}),\,(\ref{eq:pro1_3-1-1})\nonumber, \\
 & \left\Vert \boldsymbol{m}\right\Vert =1.\label{eq:SIMO2-2}
\end{align}
\end{subequations}
\end{prop}
\begin{proof}
$\forall\boldsymbol{m}$ can be written as the multiplication of its
norm and the unit direction vector, i.e., $\boldsymbol{m}=\left\Vert \boldsymbol{m}\right\Vert \frac{\boldsymbol{m}}{\left\Vert \boldsymbol{m}\right\Vert }.$
If we let $\tilde{\boldsymbol{m}}=\frac{\boldsymbol{m}}{\parallel\boldsymbol{m}\parallel}$,
the objective function of problem (\ref{eq:SIMO1}) is equivalent
to $\underset{k}{\max}\frac{\left\Vert \tilde{\boldsymbol{m}}\right\Vert ^{2}}{K^{2}P_{k}r_{k}^{2}\left\Vert \tilde{\boldsymbol{m}}^{\mathrm{H}}\boldsymbol{H}_{k}\right\Vert ^{2}},$
where $\left\Vert \tilde{\boldsymbol{m}}\right\Vert =1$. This completes
the proof.
\end{proof}

\begin{rem}
The SIMO scenario is a special case of the MIMO case, where $A_{k}=\boldsymbol{m}^{\mathrm{H}}\boldsymbol{h}_{k}b_{k}$,
$B=\boldsymbol{m}^{\mathrm{H}}\boldsymbol{n}$. The problem formulated
under SIMO case is the same as (\ref{eq:SIMO2}) except that the channel
and transmit coefficients are respectively vector $\boldsymbol{h}_{k}\in\mathbb{C}^{N_{t}}$
and scalar $b_{k}\in\mathbb{C}$.
\end{rem}

\section{DLR Optimization}\label{S:DLRsov}

In this section, we develop two algorithms to solve problems (\ref{eq:SISO2})
and (\ref{eq:SIMO2}), respectively. For the MISO scenario, problem
(\ref{eq:SISO2}) is convex and a closed-form solution is derived. For nonconvex problem (\ref{eq:SIMO2}), { we decompose the problem into two sub-problems and propose an iterative method by alternately fixing one variable and solving for the other.}

\subsection{MISO}

Obviously, problem (\ref{eq:SISO2}) is convex. For notation simplicity,
we {let} $l_{k}=\frac{1}{r_{k}}$ and $c_{k}=\frac{1}{K\sqrt{P_{k}}\left\Vert \boldsymbol{h}_{k}\right\Vert }${.
As such,} problem (\ref{eq:SISO2}) can be written as \begin{subequations}
\begin{align}
\underset{l_{k}}{\min}\quad & \max_{k}\left(c_{k}l_{k}\right)^{2}\\
\mathrm{s.t.}\quad & \sum\nolimits_{k=1}^{K}l_{k}=K,\label{eq:SIMO_C1}\\
 & \frac{1}{r_{\max}}\leq l_{k}\leq\frac{1}{r_{\min}},\quad\forall k.\label{eq:SIMO_C2}
\end{align}
\label{eq:SISO3}
\end{subequations}The solution that minimizes $\max_{k}\left(c_{k}l_{k}\right)^{2}$
also minimizes $\max_{k}c_{k}l_{k}$, which indicates that their solutions
are identical. Consequently, problem (\ref{eq:SISO3}) is further
equivalent to the following problem
\begin{equation}
\underset{l_{k}}{\min}\,\max_{k}c_{k}l_{k}\quad\mathrm{s.t.}\:(\ref{eq:SIMO_C1}),\,(\ref{eq:SIMO_C2}),
\end{equation}
which is a {typical} linear programming problem.

Assuming $k=\mathrm{arg}\max_{i}c_{i}l_{i}$, we then have $c_{i}l_{i}\leq c_{k}l_{k}$
for all $i$. Following the equation constraint (\ref{eq:SIMO_C1}),
we have
\begin{equation}
K=\sum\nolimits_{i=1}^{K}l_{i}\leq\sum\nolimits_{i=1}^{K}\frac{c_{k}l_{k}}{c_{i}}=c_{k}l_{k}\sum\nolimits_{i=1}^{K}\frac{1}{c_{i}}.\label{eq:MISO_closed_form1}
\end{equation}
Thus
\begin{equation}
c_{k}l_{k}\geq\frac{K}{\sum_{i=1}^{K}\frac{1}{c_{i}}}.\label{eq:MISO_closed_form2}
\end{equation}

\begin{thm}
\label{thm:MISO}In a MISO system, the MSE is lower bounded by $\mathrm{MSE^{lb}}=\frac{\sigma^{2}}{\left(\sum_{i=1}^{K}\sqrt{P_{i}}\left\Vert \boldsymbol{h}_{i}\right\Vert \right)^{2}}$.
Denoting the MSE obtained with and without considering DLR as $\mathrm{MSE^{d}}$
and $\mathrm{MSE^{n}}$, we always have
\begin{equation}
\mathrm{MSE^{n}}\overset{\left(a\right)}{\geq}\mathrm{MSE^{d}}\overset{\left(b\right)}{\geq}\mathrm{MSE^{lb}},\label{eq:MISO_ther}
\end{equation}
where equality $\left(a\right)$ holds if and only if $\sqrt{P_{i}}\left\Vert \boldsymbol{h}_{i}\right\Vert =\sqrt{P_{j}}\left\Vert \boldsymbol{h}_{j}\right\Vert ,\forall i,j\in\mathcal{K}$,
and equality $\left(b\right)$ holds if and only if $r_{i}\sqrt{P_{i}}\left\Vert \boldsymbol{h}_{i}\right\Vert =r_{j}\sqrt{P_{j}}\left\Vert \boldsymbol{h}_{j}\right\Vert ,\forall i,j\in\mathcal{K}$.
\end{thm}
\begin{proof}
According to (\ref{eq:MISO_closed_form2}), the lower bound of $c_{k}l_{k}$
is $\frac{K}{\sum_{i=1}^{K}\frac{1}{c_{i}}}$. Hence, the MSE is lower
bounded by
\begin{equation}
\mathrm{MSE^{lb}}=\left(\frac{K}{\sum_{i=1}^{K}\frac{1}{c_{i}}}\right)^{2}\sigma^{2}=\frac{\sigma^{2}}{\left(\sum_{i=1}^{K}\sqrt{P_{i}}\left\Vert \boldsymbol{h}_{i}\right\Vert \right)^{2}}.
\end{equation}
The lower bound is {attained if and only if} $c_{i}l_{i}=c_{j}l_{j},\forall i,j\in\mathcal{K}$, {which also means that} $r_{i}\sqrt{P_{i}}\left\Vert \boldsymbol{h}_{i}\right\Vert =r_{j}\sqrt{P_{j}}\left\Vert \boldsymbol{h}_{j}\right\Vert ,\forall i,j\in\mathcal{K}$.

With loss of generality, {we} assume that $c_{k}$ is sorted in a descending
order, i.e., $c_{i}\geq c_{j},\forall i>j$. Thus, the MSE without
considering DLR can be readily obtained as
\begin{equation}
\mathrm{MSE^{n}}=c_{1}^{2}\sigma^{2}=\frac{\sigma^{2}}{K^{2}P_{1}\left\Vert \boldsymbol{h}_{1}\right\Vert ^{2}},
\end{equation}
where $l_{k}$ can be regarded to have equal value of $1$. When taking
the DLR into consideration, we can always find a feasible set of coefficients
$\left[l_{1},l_{2},\ldots,l_{K}\right]$, which guarantee both $c_{1}\geq c_{1}l_{1}$
and $c_{1}l_{1}=\max\left(c_{i}l_{i},\forall i\in\mathcal{K}\right).$
Consequently, we have
\begin{equation}
\mathrm{MSE^{d}}=c_{1}^{2}l_{1}^{2}\sigma^{2}\leq c_{1}^{2}\sigma^{2}=\mathrm{MSE^{n}},
\end{equation}
the equality of which holds if and only if when $c_{i}=c_{j},\forall i,j\in\mathcal{K}$,
i.e., $\sqrt{P_{i}}\left\Vert \boldsymbol{h}_{i}\right\Vert =\sqrt{P_{j}}\left\Vert \boldsymbol{h}_{j}\right\Vert ,\forall i,j\in\mathcal{K}$.
Finally, we complete the proof.
\end{proof}
To further minimize the MSE, we should optimize $l_{i}$. Based on
the above analysis, the optimal solution of $l_{i}$ under constraint
(\ref{eq:SIMO_C2}) is given by
\begin{equation}
l_{i}=\mathrm{clip}\left(\frac{c_{k}l_{k}}{c_{i}},\left[\frac{1}{r_{\max}},\frac{1}{r_{\min}}\right]\right),\label{eq:closed-form_l}
\end{equation}
where
\begin{equation}
\sum_{i=1}^{K}l_{i}=K.
\end{equation}
{Operation $\mathrm{clip}\left(x,\left[a,b\right]\right)$ truncates $x$ to the specified interval $\left[a,b\right]$.}

The overall procedure to solve problem (24) is shown in \textbf{Algorithm}
\ref{AlgorithmMISO}. According to (\ref{eq:closed-form_l}), the
proposed scheme needs to know the user index $k$ with the maximum
value $c_{k}l_{k}$. To find the device index, we exhaustively search
all devices, which indicates that the number of iterations in the
outer layer is $K$. For a given device index $k$, {the bisection technique is applied to find a solution $l_k$, the complexity of which is $\mathcal{O}(\log_{2}(1/\delta))$
with accuracy $\delta$.}

\begin{algorithm}[t] 	
\caption{\label{AlgorithmMISO} Optimal Learning Rate for MISO.}
\textbf{Input}: $c_{k}$, ${\mathrm{obj}}=\frac{\max_k(c_k)}{r_{\min}}$, $\delta$, $\mathrm{Num}=20$\\
\textbf{Output}: $l_i^{\mathrm{opt}}, {\mathrm{obj}}$\\
\textbf{Initialize}: $l_k^\mathrm{min}=\frac{1}{r_{\max}}$, $l_k^\mathrm{max}=\frac{1}{r_{\min}}$	
\begin{algorithmic}[1]
\State \textbf{for} $k=1:K$
\State \quad \textbf{while} $|\sum_{i=1}^{K} l_{i}- K|\leq \delta$
\State \quad \quad $l_k=(l_k^\mathrm{max}+l_k^\mathrm{min})/2$
\State \quad \quad $l_i = \mathrm{clip}\left(\frac{c_{k}l_{k}}{c_{i}},[\frac{1}{r_{\max}},\frac{1}{r_{\min}}]\right),\forall i\in\mathcal{K}$
\State \quad \quad  \textbf{if} ($\sum_{i=1}^{K} l_{i}\geq K$)
\State \quad \quad \quad $l_k^{\mathrm {max}}=l_k$
\State \quad \quad  \textbf{else}
\State \quad \quad \quad $l_k^{\mathrm {min}}=l_k$
\State \quad \textbf{if} ${\mathrm{obj}} \geq \max(c_il_i)$
\State \quad \quad ${\mathrm{obj}}=\max(c_il_i)$
\State \quad \quad $l_i^{\mathrm{opt}}=l_i, \forall i\in\mathcal{K}$
\end{algorithmic}
\end{algorithm}

\subsection{MIMO}

{Suppose that $c_{k}=\frac{\left\Vert \boldsymbol{m}\right\Vert }{K\sqrt{P_{k}}\left\Vert \boldsymbol{m}^{\mathrm{H}}\boldsymbol{H}_{k}\right\Vert }$,
$l_{k}=\frac{1}{r_{k}}$, and} $c_{k}l_{k}=\max_{i}\left(c_{i}l_{i}\right)$,
we have the following theorem.
\begin{thm}
{\label{thm:MIMO}}In a MIMO system, the MSE is lower
bounded by $\mathrm{MSE^{lbm}}=\frac{\sigma^{2}}{\left(\sum_{i=1}^{K}\sqrt{P_{i}}\left\Vert \boldsymbol{m}^{\mathrm{H}}\boldsymbol{H}_{i}\right\Vert \right)^{2}}$
for any given $\boldsymbol{m}$, and we always have
\begin{equation}
\mathrm{MSE^{n}}\overset{\left(a\right)}{\geq}\mathrm{MSE^{d}}\overset{\left(b\right)}{\geq}\mathrm{MSE^{lbm}},\label{eq:MIMO_ther}
\end{equation}
where {the equalities} of $\left(a\right)$ and $\left(b\right)$
hold when $\sqrt{P_{i}}\left\Vert \boldsymbol{m}^{\mathrm{H}}\boldsymbol{H}_{i}\right\Vert =\sqrt{P_{j}}\left\Vert \boldsymbol{m}^{\mathrm{H}}\boldsymbol{H}_{j}\right\Vert ,\forall i,j\in\mathcal{K}$, and
$r_{i}\sqrt{P_{i}}\left\Vert \boldsymbol{m}^{\mathrm{H}}\boldsymbol{H}_{i}\right\Vert =r_{j}\sqrt{P_{j}}\left\Vert \boldsymbol{m}^{\mathrm{H}}\boldsymbol{H}_{j}\right\Vert ,\forall i,j\in\mathcal{K}$,
 respectively.
\end{thm}
\begin{proof}
According to the equality constraints (\ref{eq:pro1-2-1-1}) and (\ref{eq:SIMO2-2}),
we arrive at
\begin{equation}
c_{k}l_{k}\!\geq\!\frac{K}{\sum\limits_{i=1}^{K}\displaystyle\frac{1}{c_{i}}}\!=\!\frac{\left\Vert \boldsymbol{m}\right\Vert }{\sum\limits_{i=1}^{K}\!\sqrt{P_{i}}\left\Vert \boldsymbol{m}^{\mathrm{H}}\boldsymbol{H}_{i}\right\Vert }\!=\!\frac{1}{\sum\limits_{i=1}^{K}\!\sqrt{P_{i}}\left\Vert \boldsymbol{m}^{\mathrm{H}}\boldsymbol{H}_{i}\right\Vert }.
\end{equation}
Thus, $\mathrm{MSE^{lbm}}=\frac{\sigma^{2}}{\left(\sum_{i=1}^{K}\sqrt{P_{i}}\left\Vert \boldsymbol{m}^{\mathrm{H}}\boldsymbol{H}_{i}\right\Vert \right)^{2}}$
is the lower bound of MSE for any given feasible $\boldsymbol{m}$,
which is achieved only {if} $c_{i}l_{i}=c_{j}l_{j},\forall i,j\in\mathcal{K}$,
i.e., $r_{i}\sqrt{P_{i}}\left\Vert \boldsymbol{m}^{\mathrm{H}}\boldsymbol{H}_{i}\right\Vert =r_{j}\sqrt{P_{j}}\left\Vert \boldsymbol{m}^{\mathrm{H}}\boldsymbol{H}_{j}\right\Vert ,\forall i,j\in\mathcal{K}$.

{We first put aside the DLR and denote the equivalent channel as} $\boldsymbol{h}'_{i}=\boldsymbol{m}^{\mathrm{H}}\sqrt{P_{i}}\boldsymbol{H}_{i}\in\mathbb{C}^{N_{d}}$. {As such,}
the problem of MSE minimization {becomes to find} $\boldsymbol{m}$
that maximizes the minimum $\ell_{2}$-norm of $\boldsymbol{h}'_{i}$.
Denoting its optimal solution as $\boldsymbol{m}^{*}$, the corresponding
minimum $\ell_{2}$-norm and MSE can be expressed as $h_{\mathrm{min}}^{\mathrm{norm}}=\min\left(\sqrt{P_{i}}\left\Vert \boldsymbol{\left(m^{*}\right)}^{\mathrm{H}}\boldsymbol{H}_{i}\right\Vert \right)$,
and $\mathrm{MSE}^{\mathrm{n}}=\frac{\sigma^{2}}{\left(Kh_{\mathrm{min}}^{\mathrm{norm}}\right)^{2}}$.
Then, {we} consider the DLR in the following two cases.

Case 1: {The} lower bound is not {attained}, i.e., $\mathrm{MSE}^{\mathrm{n}}>\mathrm{MSE}^{\mathrm{lbm}}$.
Without loss of generality, assume that $c_{k}$ is sorted in an ascending
order, i.e., $\left\Vert \boldsymbol{h}_{i}'\right\Vert \geq\left\Vert \boldsymbol{h}_{j}'\right\Vert ,\forall i>j$.
There always exists a feasible set of DLR coefficients $\left[r_{1},r_{2},\ldots,r_{K}\right]$
such that $r_{K}\left\Vert \boldsymbol{h}_{K}'\right\Vert >\left\Vert \boldsymbol{h}_{K}'\right\Vert $
with $r_{K}\left\Vert \boldsymbol{h}_{K}'\right\Vert =\min\left(r_{i}\left\Vert \boldsymbol{h}_{i}'\right\Vert \right)$.
Thus, $\mathrm{MSE}^{\mathrm{n}}>\mathrm{MSE}^{\mathrm{d}}\overset{\left(b\right)}{\geq}\mathrm{MSE^{lbm}}$,
and the equality of $\left(b\right)$ holds only when $r_{i}\left\Vert \boldsymbol{h}_{i}'\right\Vert =r_{j}\left\Vert \boldsymbol{h}_{j}'\right\Vert ,\forall i,j\in\mathcal{K}$.

Case 2: {The} lower bound is achieved i.e., $\mathrm{MSE}^{\mathrm{n}}=\mathrm{MSE}^{\mathrm{lbm}}$.
In this case, we have $\left\Vert \boldsymbol{h}_{i}\right\Vert =\left\Vert \boldsymbol{h}_{j}\right\Vert ,\forall i,j\in\mathcal{K}$
and $\mathrm{MSE}^{\mathrm{n}}=\mathrm{MSE}^{\mathrm{d}}=\mathrm{MSE}^{\mathrm{lbm}}$.
The introduction of DLR {cannot} further improve the performance, and DLR {in this case is equal to} $1$.

Therefore, we complete the proof.
\end{proof}

In the MIMO scenario, problem (\ref{eq:SIMO2}) is difficult due to
the nonconvex objective (\ref{eq:SIMO2-1}) and constraint (\ref{eq:SIMO2-2}).
By introducing an auxiliary variable $\tau$, problem (\ref{eq:SIMO2})
is further equivalent to the following problem:
\begin{subequations}\label{eq:SIMO3}
\begin{align}
\min_{\boldsymbol{m},r_{k},\tau}\quad & \tau\label{eq:SIMO3-1}\\
\mathrm{s.t.\quad} & \left\Vert \boldsymbol{m}\right\Vert ^{2}\leq\tau K^{2}P_{k}\left\Vert \boldsymbol{m}^{\mathrm{H}}\boldsymbol{H}_{k}\right\Vert ^{2}r_{k}^{2},\quad\forall k,\label{eq:SIMO3-2}\\
 & (\ref{eq:pro1-2-1-1}),(\ref{eq:pro1_3-1-1}),(\ref{eq:SIMO2-2}).\nonumber
\end{align}
\end{subequations}
To solve problem (\ref{eq:SIMO3}), we utilize the iterative technique
and decompose the problem into two sub-problems by {alternately} fixing the DLR
ratio $r_{k}$ and the receive beamforming vector $\boldsymbol{m}$, respectively.

Given the DLR ratio $r_{k}$, the original problem (\ref{eq:SIMO3})
is reduced into the following sub-problem:
\begin{equation}
\min_{\boldsymbol{m},\tau}\;\tau\qquad\mathrm{s.t.}\quad(\ref{eq:SIMO2-2}),(\ref{eq:SIMO3-2}),\label{eq:SIMO3-A}
\end{equation}
which is nonconvex due to constraints (\ref{eq:SIMO2-2}) and (\ref{eq:SIMO3-2}).
To solve problem (\ref{eq:SIMO3-A}), we have the following lemma.
\begin{lem}
\label{thm:thm2}{Suppose that we have a} semidefinite matrix $\boldsymbol{A}_{k}=\boldsymbol{H}_{k}\boldsymbol{H}_{k}^{\mathrm{H}}$,
if $\det\boldsymbol{A}_{k}>0$, the range of $\frac{\left\Vert \boldsymbol{m}\right\Vert ^{2}}{K^{2}P_{k}\left\Vert \boldsymbol{m}^{\mathrm{H}}\boldsymbol{H}_{k}\right\Vert ^{2}r_{k}^{2}}$
is $\left[\frac{1}{K^{2}P_{k}r_{k}^{2}\lambda_{\mathrm{max}}},\frac{1}{K^{2}P_{k}r_{k}^{2}\lambda_{k,\mathrm{min}}}\right]$;
otherwise $\left[\frac{1}{K^{2}P_{k}r_{k}^{2}\lambda_{k,\mathrm{max}}},\infty\right]$,
where $\lambda_{k,\mathrm{max}}$ and $\lambda_{k,\mathrm{min}}$
are the maximum and minimum eigenvalues of $\boldsymbol{A}_{k}${, respectively}.
\end{lem}
\begin{proof}
Define a function $g\left(\boldsymbol{m}\right)=\frac{\parallel\boldsymbol{m}^{\mathrm{H}}\boldsymbol{H}\parallel^{2}}{\parallel\boldsymbol{m}\parallel^{2}}=\frac{\boldsymbol{m}^{\mathrm{H}}\boldsymbol{H}\boldsymbol{H}^{\mathrm{H}}\boldsymbol{m}}{\boldsymbol{m}^{\mathrm{H}}\boldsymbol{m}},$
which is the Rayleigh-Ritz of matrix $\boldsymbol{A}=\boldsymbol{H}\boldsymbol{H}^{\mathrm{H}}$.
{Suppose that} the maximum and minimum eigenvalues of $\boldsymbol{A}$
are {respectively} $\lambda_{\mathrm{max}}$ and $\lambda_{\mathrm{min}}$, function
$g\left(\boldsymbol{m}\right)$ is within the interval of $\left[\lambda_{\mathrm{min}},\lambda_{\mathrm{max}}\right]$.
Apparently, $\boldsymbol{A}\left(\neq\boldsymbol{0}\right)$ is semidefinite
and we have $\det\boldsymbol{A}\geq0$. In the case of $\det\boldsymbol{A}>0$,
the eigenvalues of $\boldsymbol{A}$ are positive. Consequently, we
have
\begin{equation}
\frac{1}{\lambda_{\mathrm{max}}}\leq\frac{1}{g\left(\boldsymbol{m}\right)}=\frac{\parallel\boldsymbol{m}\parallel^{2}}{\parallel\boldsymbol{m}^{\mathrm{H}}\boldsymbol{H}\parallel^{2}}\leq\frac{1}{\lambda_{\mathrm{min}}}.
\end{equation}
Otherwise, $\lambda_{\mathrm{min}}=0$, we have
\begin{equation}
\frac{1}{\lambda_{\mathrm{max}}}\leq\frac{1}{g\left(\boldsymbol{m}\right)}=\frac{\parallel\boldsymbol{m}\parallel^{2}}{\parallel\boldsymbol{m}^{\mathrm{H}}\boldsymbol{H}\parallel^{2}}\leq\infty.
\end{equation}
Therefore, the proof is complete.
\end{proof}
\begin{rem}
{Let} $\tau_{k}^{l\mathrm{ow}}=\frac{1}{K^{2}P_{k}r_{k}^{2}\lambda_{k,\mathrm{max}}},$
$\tau_{k}^{\mathrm{up}}=\frac{1}{K^{2}P_{k}r_{k}^{2}\lambda_{k,\mathrm{min}}}$
if $\lambda_{\mathrm{min}}>0$; otherwise{{} }$\tau_{k}^{\mathrm{up}}=\infty$,
and then $\frac{\left\Vert \boldsymbol{m}\right\Vert ^{2}}{K^{2}P_{k}\left\Vert \boldsymbol{m}^{\mathrm{H}}\boldsymbol{H}_{k}\right\Vert ^{2}r_{k}^{2}}\in\left[\tau_{k}^{\mathrm{low}},\tau_{k}^{\mathrm{up}}\right]$.
The necessary condition of $\tau$ that guarantees the feasibility
of problem (\ref{eq:SIMO3-A}) is $\text{\ensuremath{\tau\in}}\left[\tau^{\mathrm{low}},\tau^{\mathrm{up}}\right]$,
where
\begin{equation}
\tau^{\mathrm{low}}=\min_{k}\left(\tau_{k}^{\mathrm{low}}\right),\:\tau^{\mathrm{up}}=\max_{k}\left(\tau_{k}^{\mathrm{up}}\right).\label{eq:tau_low}
\end{equation}
\end{rem}
  For any given $\tau\in\left[\tau^{\mathrm{low}},\tau^{\mathrm{up}}\right]$,
sub-problem (\ref{eq:SIMO3-A}) is interpreted {as finding} the receive
beamforming vector $\boldsymbol{m}$ that makes the sub-problem feasible,
which is a check problem. By introducing $\boldsymbol{M}=\boldsymbol{m}\boldsymbol{m}^{\mathrm{H}}$,
the sub-problem given $\tau$ is converted {to}
\begin{subequations}\label{eq:SIMO4-A}
\begin{align}
 \min_{\boldsymbol{M}} \quad & 0\label{eq:pro3A1-1}\\
\mathrm{s.t.\quad} & \mathrm{Tr}\left(\boldsymbol{M}\right)\leq\tau K^{2}P_{k}\mathrm{Tr}\left(\boldsymbol{M}\boldsymbol{H}_{k}\boldsymbol{H}_{k}^{\mathrm{H}}\right)r_{k}^{2},\quad\forall k,\label{eq:pro3A2-1}\\
 & \boldsymbol{M}\succcurlyeq0\\
 & \mathrm{Tr}\left(\boldsymbol{M}\right)=1\label{eq:pro3A3-1}\\
 & \mathrm{rank}\left(\boldsymbol{M}\right)=1.\label{eq:pro3A4-1}
\end{align}

\end{subequations}

Indeed, the only difficulty of the above problem lies in the rank
one constraint (\ref{eq:pro3A4-1}). The problem can be solved by
first dropping {the} constraint (\ref{eq:pro3A4-1}) to obtain solution
$\boldsymbol{M}^{*}$. Then, the receive beamforming vector $\boldsymbol{m}^{*}$
can be calculated using the eigenvector approximation method or the
randomization technique \cite{5447068}, which is sub-optimal especially
when $\boldsymbol{M}$ is large. To guarantee the rank one constraint,
we utilize a DC representation \cite{kaiyang,2018DC}, and {convert}
problem (\ref{eq:SIMO4-A}) {to}
\begin{equation}
\min_{\boldsymbol{M}}\;\mathrm{Tr}\left(\boldsymbol{M}\right)-\left\Vert \boldsymbol{M}\right\Vert _{2}\quad\mathrm{s.t.}\;(\ref{eq:pro3A1-1}),(\ref{eq:pro3A2-1}),(\ref{eq:pro3A3-1}),(\ref{eq:pro3A4-1}),\label{eq:SIMO4-AA}
\end{equation}
which can be efficiently solved by DC programming with complexity
of $\mathcal{O}\left(N_{t}^{3}\right)$.

To find {the} solution $\tau$, we utilize the bisection technique. Specifically,
the interval $\left[\tau^{\mathrm{low}},\tau^{\mathrm{up}}\right]$
is divided into two sub-intervals $\left[\tau^{\mathrm{low}},\tau\right]$
and $\left[\tau,\tau^{\mathrm{up}}\right]$ with $\tau=\left(\tau^{\mathrm{low}}+\tau^{\mathrm{up}}\right)/2.$
If problem (\ref{eq:SIMO4-AA}) is solved for $\tau$, then the solution
is within $\left[\tau^{\mathrm{low}},\tau\right]$; otherwise we have
$\left[\tau,\tau^{\mathrm{up}}\right]$. {Repeatedly} checking the problem
until $\tau^{\mathrm{up}}-\tau^{\mathrm{low}}<\delta$, where $\delta$
is the accuracy. Such bisection technique involves $\log_{2}\frac{1}{\delta}$
{repetitive operations and, hence,} solving {the} problem (\ref{eq:SIMO3-A}) requires the
computational complexity of $\mathcal{O}\left(N_{t}^{3}\log_{2}\frac{1}{\delta}\right)$.
The algorithm is summarized in \textbf{Algorithm}~\ref{AlgorithmSubA}.

\begin{algorithm}[!hptb] 	
\caption{\label{AlgorithmSubA} Beamforming design for MIMO.}
\textbf{Input}: $\boldsymbol{H}^{\mathrm{H}}_{k}$, $r_{k}$, $\delta$\\
\textbf{Output}: $\boldsymbol{m}$, $\tau$	
\begin{algorithmic}[1]
\State calculate $\tau^{\mathrm{low}}$,  $\tau^{\mathrm{up}}$ based on (\ref{eq:tau_low})		
\State \textbf{while} $(\tau^{\mathrm{up}}-\tau^{\mathrm{low}})>\delta$
\State \quad \textbf{if} problem  (\ref{eq:SIMO4-AA}) is infeasible
\State \qquad $\tau^{\mathrm{up}}=(\tau^{\mathrm{up}}+\tau^{\mathrm{low}})/2$
\State \quad \textbf{else}
\State \qquad $\tau^{\mathrm{low}}=(\tau^{\mathrm{up}}+\tau^{\mathrm{low}})/2$	
\end{algorithmic}
\end{algorithm}

Given the obtained receive beamforming vector $\boldsymbol{m}$, the
original problem (\ref{eq:SIMO3}) is {reduced to the sub-problem below:}
\begin{equation}
\min_{r_{k},\tau}\quad\tau\qquad\mathrm{s.t.}\quad(\ref{eq:pro1-2-1-1}),(\ref{eq:pro1_3-1-1}),(\ref{eq:SIMO3-2}).\label{eq:SIMO3-B}
\end{equation}

Denoting the equivalent channel vector $\boldsymbol{h}'_{k}$ as $\boldsymbol{m}^{\mathrm{H}}\boldsymbol{H}_{k}$, {the}
sub-problem (\ref{eq:SIMO3-B}) can be transformed into the problem
under the MISO case with channel vector $\boldsymbol{h}'_{k}$. Note
that the sub-problem under the SIMO case is equivalent to that of
the SISO case by defining equivalent channel coefficient $h'_{k}$
as $\boldsymbol{m}^{\mathrm{H}}\boldsymbol{h}_{k}$. Therefore, we
can readily {derive the closed-form solution in} the MISO case.
{We let} $l_{k}=\frac{1}{r_{k}}$, $c_{k}=\frac{1}{K\sqrt{P_{k}}\left\Vert \boldsymbol{h}_{k}^{'}\right\Vert }=\frac{1}{K\sqrt{P_{k}}\left\Vert \boldsymbol{m}^{\mathrm{H}}\boldsymbol{H}_{k}\right\Vert }$, {and then}
the solution is given by $l_{i}=\mathrm{clip}\left(\frac{c_{k}l_{k}}{c_{i}},[\frac{1}{r_{\max}},\frac{1}{r_{\min}}]\right)$,
where $\sum_{i=1}^{K}l_{i}=K$.

Thus, the sub-optimal solution of problem (\ref{eq:SIMO3}) can be
obtained by alternatively solving sub-problems (\ref{eq:SIMO4-AA})
and (\ref{eq:SIMO3-B}). For each iteration, the computational complexity
is $\mathcal{O}\left(N_{t}^{3}\log_{2}\frac{1}{\delta_{1}}+K\log_{2}\frac{1}{\delta_{2}}\right)$,
where $\delta_{1}$ and $\delta_{2}$ are the accuracy of solving
receive beamforming vector $\boldsymbol{m}$ and DLR ratio $r_{k}$,
respectively. The iterative method proposed is summarized in\textbf{
Algorithm}~\ref{AlgorithmMIMO}.
\begin{algorithm}[!htpb] 	
\caption{\label{AlgorithmMIMO} Iterative Learning Rate and Receive Beamforming.}
\textbf{Input}: $\boldsymbol{H}_{k}$\\
\textbf{Output}: $\boldsymbol{m}$, $r_{k}$, $\tau$\\
\textbf{Initialize}: $r_k=1$
\begin{algorithmic}[1]
\State \textbf{do} loop
\State \quad given $r_k$, solve problem (\ref{eq:SIMO3-A})  using \textbf{Algorithm} \ref{AlgorithmSubA}

\State \quad given $\boldsymbol{m}$, solve problem (\ref{eq:SIMO3-B})  using \textbf{Algorithm} \ref{AlgorithmMISO}

\State \textbf{until} $\tau$ converges
	
\end{algorithmic}
\end{algorithm}

\section{Asymptotic Analysis and {Receive Beamforming Design}}\label{S:AAna}

This section presents the theoretical analysis of {the MSE and the DLR ratio in} the MISO, SIMO,
and MIMO scenarios, {when} the numbers of {antennas} $N_{d}$ and $N_{t}$ {increase to} infinity. Based on the asymptotic analysis, we propose a {near-optimal} and closed-form receive beamforming solution. Note that each device is assumed to have equal maximum power $P_{k}=P$.

\subsection{MISO}

In the MISO case, we present asymptotic analysis when {the number of antennas} $N_{d}$ at the devices goes {to} infinity. Since the channels
between the devices and the aggregator are assumed to be independently
Rayleigh distributed, we have
\begin{equation}
\left\Vert \boldsymbol{h}_{k}\right\Vert \rightarrow \sqrt{N_{d}},
\end{equation}
%\,\,\mathop{ = }\limits^{{N_d} \to \infty} \,\,
\begin{equation}
c_{k}=\frac{1}{K\sqrt{P}\left\Vert \boldsymbol{h}_{k}\right\Vert }\rightarrow\frac{1}{K\sqrt{P}\sqrt{N_{d}}},
\end{equation}
which {suggest} that $c_{i}=c_{j},\forall i,j\in\mathcal{K}$. As a
consequence, lower bound $\mathrm{MSE^{lb}}$ can be achieved according
to \textbf{Theorem} \ref{thm:MISO} and the {equalities in} both (a) and
(b) {of} (\ref{eq:MISO_ther}) {are} guaranteed. Accordingly, the
MSE and $r_{k}$ {become}
\begin{equation}
\mathrm{MSE}=\left(\frac{1}{\sqrt{P}\sum_{i=1}^{K}\left\Vert \boldsymbol{h}_{i}\right\Vert }\right)^{2}\sigma^{2}\rightarrow\frac{\sigma^{2}}{PK^{2}N_{d}},
\end{equation}
\begin{equation}
r_{k}=\frac{1}{l_{k}}=c_{k}\sqrt{P}\sum\nolimits_{i=1}^{K}\left\Vert \boldsymbol{h}_{i}\right\Vert =\frac{\sum_{i=1}^{K}\left\Vert \boldsymbol{h}_{i}\right\Vert }{K\left\Vert \boldsymbol{h}_{k}\right\Vert }\rightarrow1,
\end{equation}
where the achieved MSE is inversely proportional to $K^{2}$ and $N_{d}$.

\subsection{SIMO}

In the SIMO case, we {let} $h'_{k}=\boldsymbol{m}^{\mathrm{H}}\boldsymbol{h}_{k}$,
$c_{k}=\frac{1}{K\sqrt{P}\left\Vert h_{k}'\right\Vert }$ and $l_{k}=\frac{1}{r_{k}}$.
{With} the increase {number of antennas} $N_{t}$, the channels between the
devices and the BS become asymptotically orthogonal{, i.e.,}
\begin{equation}
\left\langle \boldsymbol{h}_{i},\boldsymbol{h}_{j}\right\rangle \rightarrow \begin{cases}
N_{t} & i=j\\
0 & i\neq j
\end{cases}.\label{eq:property}
\end{equation}
By exploiting the {above property, the receive beamforming vector $\boldsymbol{m}$ can be designed simply as}
\begin{equation}
\boldsymbol{m}=\frac{\Sigma_{k}^{K} \left(\boldsymbol{h}_{k}/\left\Vert \boldsymbol{h}_{k}\right\Vert\right)}{\left\Vert\Sigma_{k}^{K} \left(\boldsymbol{h}_{k}/\left\Vert \boldsymbol{h}_{k}\right\Vert\right)\right\Vert}.\label{eq:RBF1}
\end{equation}
%\begin{equation}
%\boldsymbol{m}=\frac{\boldsymbol{h}_{1}/\left\Vert \boldsymbol{h}_{1}\right\Vert +\boldsymbol{h}_{2}/\left\Vert \boldsymbol{h}_{2}\right\Vert +\ldots+\boldsymbol{h}_{K}/\left\Vert \boldsymbol{h}_{K}\right\Vert }{\left\Vert \boldsymbol{h}_{1}/\left\Vert \boldsymbol{h}_{1}\right\Vert +\boldsymbol{h}_{2}/\left\Vert \boldsymbol{h}_{2}\right\Vert +\ldots+\boldsymbol{h}_{K}/\left\Vert \boldsymbol{h}_{K}\right\Vert \right\Vert }.\label{eq:RBF1}
%\end{equation}
Accordingly, the equivalent channel coefficient $h_{i}'$ can
be {rewritten} as
\begin{equation}
\begin{array}{*{20}{l}}
h_{i}'\!=\!\boldsymbol{m}^{\mathrm{H}}\boldsymbol{h}_{i}
\displaystyle\!=\!\frac{\Sigma_{k}^{K} \left(\boldsymbol{h}_{k}^{\mathrm H}/\left\Vert \boldsymbol{h}_{k}\right\Vert\right)}{\left\Vert\Sigma_{k}^{K} \left(\boldsymbol{h}_{k}/\left\Vert \boldsymbol{h}_{k}\right\Vert\right)\right\Vert}\boldsymbol{h}_{i}
\rightarrow \displaystyle\sqrt{\frac{N_{t}}{K}},
\end{array}
\end{equation}
which implies that $\left\Vert h_{i}'\right\Vert =\left\Vert h_{j}'\right\Vert ,\forall i,j\in\mathcal{K}$.
Therefore, the {equalities in (a) and (b) of} (\ref{eq:MIMO_ther}) can
be guaranteed and $\mathrm{MSE^{lb}}$ can be achieved according to {the} \textbf{Theorem} \ref{thm:MIMO}. In this case, the MSE and $r_{k}$ {is derived as}
\begin{align}
\mathrm{MSE} & =\left(\frac{1}{\sqrt{P}\sum_{i=1}^{K}\left\Vert h_{i}'\right\Vert }\right)^{2}\sigma^{2}\rightarrow\frac{\sigma^{2}}{PKN_{t}},
\end{align}
\begin{equation}
r_{k}=\frac{1}{l_{k}}=c_{k}\sqrt{P}\sum_{i=1}^{K}\left\Vert h_{i}'\right\Vert =\frac{\sum_{i=1}^{K}\left\Vert h_{i}'\right\Vert }{K\left\Vert h_{k}'\right\Vert }\rightarrow\frac{K\sqrt{\frac{N_{t}}{K}}}{K\sqrt{\frac{N_{t}}{K}}}=1,
\end{equation}
where the achieved MSE is inversely proportional to $K$ and $N_{t}$.

\subsection{MIMO}

In the MIMO case, both devices and the aggregator have multiple antennas.
This sub-section presents the analysis when $N_{d}$ and $N_{t}$
go to infinity, respectively. The equivalent channel vector between
device $k$ and the aggregator is denoted as $\boldsymbol{h}'_{k}=\boldsymbol{m}^{\mathrm{H}}\boldsymbol{H}_{k}$,
where $\boldsymbol{H}_{k}\in\mathbb{C}^{N_{t}\times N_{d}}$.  Since
$\boldsymbol{H}_{i}\boldsymbol{H}_{i}^{\mathrm{H}}$ is semidefinite
and $\left\Vert \boldsymbol{m}\right\Vert =1$, the norm of $\boldsymbol{h}_{i}'$
satisfies
\begin{equation}
\left\Vert \boldsymbol{h}_{i}'\right\Vert =\sqrt{\boldsymbol{m}^{\mathrm{H}}\boldsymbol{H}_{i}\boldsymbol{H}_{i}^{\mathrm{H}}\boldsymbol{m}}=\sqrt{\frac{\boldsymbol{m}^{\mathrm{H}}\boldsymbol{H}_{i}\boldsymbol{H}_{i}^{\mathrm{H}}\boldsymbol{m}}{\boldsymbol{m}^{\mathrm{H}}\boldsymbol{m}}}.
\end{equation}
{According to the Rayleigh-Ritz property, the value of $\left\Vert \boldsymbol{h}_{i}'\right\Vert$ is within the range of $\left[\sqrt{\lambda_{i,\mathrm{min}}},\sqrt{\lambda_{i,\max}}\right]$,
where $\lambda_{i,\mathrm{min}}$ and $\lambda_{i,\mathrm{max}}$
are respectively} the minimum and maximum eigenvalues of $\boldsymbol{H}_{i}\boldsymbol{H}_{i}^{\mathrm{H}}$.

First, we give the analysis when $N_{d}\rightarrow\infty$ for $N_{d}>N_{t}$.
{In} this case, the channel between  device $i$ and each antenna
at the BS is asymptotically orthogonal. Consequently, we have
\begin{equation}
\boldsymbol{H}_{i}\boldsymbol{H}_{i}^{\mathrm{H}} \rightarrow N_{d}\boldsymbol{I}_{N_{t}\times N_{t}},
\end{equation}
whose eigenvalues share the same {values} of $\sqrt{N_{d}}$. Thus, according
to the Rayleigh-Ritz property, for any beamforming vector $\boldsymbol{m}$
with $\left\Vert \boldsymbol{m}\right\Vert =1$, we have
\begin{equation}
\left\Vert \boldsymbol{h}_{1}'\right\Vert =\left\Vert \boldsymbol{h}_{2}'\right\Vert =\ldots=\left\Vert \boldsymbol{h}_{K}'\right\Vert \rightarrow\sqrt{N_{d}},
\end{equation}
which implies that the condition of the equalities of (a) and (b) in (\ref{eq:MIMO_ther})
are guaranteed. Based on \textbf{Theorem }\ref{thm:MIMO}, the MSE
and $r_{k}$ are obtained such that
\begin{align}
\mathrm{MSE} & =\left(\frac{1}{\sqrt{P}\sum_{i=1}^{K}\left\Vert \boldsymbol{h}_{i}'\right\Vert }\right)^{2}\sigma^{2}\rightarrow\frac{\sigma^{2}}{PK^{2}N_{d}},
\end{align}
\begin{equation}
r_{k}=\frac{1}{l_{k}}=c_{k}\sqrt{P}\sum_{i=1}^{K}\left\Vert \boldsymbol{h}_{i}'\right\Vert =\frac{\sum_{i=1}^{K}\left\Vert \boldsymbol{h}_{i}\right\Vert }{K\left\Vert \boldsymbol{h}_{k}\right\Vert }\rightarrow1.
\end{equation}
Note that the {MSE achieved} is inversely proportional to $K^{2}$ and
$N_{d}$, irrespective of $N_{t}$ when $N_{d}\rightarrow\infty$ for
$N_{d}>N_{t}$. The reason is that $\boldsymbol{H}_{i}$ has
full row rank with equal singular values of $\sqrt{N_{d}}$, which
indicates that the equivalent channel $\boldsymbol{h}_{i}^{'}$ has
the same power of $\sqrt{N_{d}}$ {regardless} of $N_{t}$.

Next, the analysis of the MSE {when $N_{t}\rightarrow\infty$ for
$N_{t}>N_{d}$ is provided}. The channels between each antenna of device $i$ and
the BS are asymptotically orthogonal with power $N_{t}$. Thus, the
rank of $\boldsymbol{H}_{i}$ is $r=\mathrm{rank}\left(\boldsymbol{H}_{i}\right)=N_{d}$.
By using singular value decomposition (SVD), {we readily have}
\begin{equation}
\boldsymbol{H}_{i}\boldsymbol{H}_{i}^{\mathrm{H}}\!=\!\boldsymbol{U}\boldsymbol{\Sigma}^{2}\boldsymbol{U}^{\mathrm{H}} \rightarrow \boldsymbol{U}\!\left[\begin{array}{cc}
N_{t}\boldsymbol{I}_{N_{d}\times N_{d}} & \boldsymbol{0}\\
\boldsymbol{0} & \boldsymbol{0}
\end{array}\right]\!\boldsymbol{U}^{\mathrm{H}},
\end{equation}
where the first $N_{d}$ columns of $\boldsymbol{U}$ are $\boldsymbol{U}_{r}=\left[\frac{\boldsymbol{H}_{i}\left[:,1\right]}{\left\Vert \boldsymbol{H}_{i}\left[:,1\right]\right\Vert },\frac{\boldsymbol{H}_{i}\left[:,2\right]}{\left\Vert \boldsymbol{H}_{i}\left[:,2\right]\right\Vert },\ldots,\frac{\boldsymbol{H}_{i}\left[:,N_{d}\right]}{\left\Vert \boldsymbol{H}_{i}\left[:,N_{d}\right]\right\Vert }\right]$.

Thus, the minimum and maximum eigenvalues of $\boldsymbol{H}_{i}\boldsymbol{H}_{i}^{\mathrm{H}}$
are {respectively} $\lambda_{i,\mathrm{min}}=0$, $\lambda_{i,\mathrm{max}}=N_{t}$.
Besides, the asymptotically orthogonal property {dictates that} the column
spaces spanned by $\boldsymbol{H}_{k},\forall k\in\mathcal{K}$ are
orthogonal, i.e.,
\begin{equation}
\mathrm{span}\left(\boldsymbol{H}_{i}\right)\bot\mathrm{span}\left(\boldsymbol{H}_{j}\right),\:\forall i\neq j,
\end{equation}
which suggests that $\boldsymbol{h}_{i}\perp\boldsymbol{h}_{j}$ for any $\boldsymbol{h}_{i} \in \mathrm{span}\left(\boldsymbol{H}_{i}\right)$ and $ \boldsymbol{h}_{j} \in \mathrm{span}\left(\boldsymbol{H}_{j}\right)$.

Therefore, the receive beamforming vector can be designed in a simple
manner as
\begin{equation}
\boldsymbol{m}=\frac{\widetilde{\boldsymbol{h}}_{1}+\widetilde{\boldsymbol{h}}_{2}+\ldots+\widetilde{\boldsymbol{h}}_{K}}{\left\Vert \widetilde{\boldsymbol{h}}_{1}+\widetilde{\boldsymbol{h}}_{2}+\ldots+\widetilde{\boldsymbol{h}}_{K}\right\Vert },
\end{equation}
where $\widetilde{\boldsymbol{h}}_{i}$ is the eigenvector of $\boldsymbol{H}_{i}\boldsymbol{H}_{i}^{\mathrm{H}}$,
which can be any column vector of $\boldsymbol{U}_{r}$. Further,
the equivalent channels {can be expressed as}
\begin{align}
\boldsymbol{h}_{i}' & =\boldsymbol{m}^{\mathrm{H}}\boldsymbol{H}_{i}=\frac{\left(\widetilde{\boldsymbol{h}}_{1}+\widetilde{\boldsymbol{h}}_{2}+\ldots+\widetilde{\boldsymbol{h}}_{K}\right)^{\mathrm{H}}}{\left\Vert \widetilde{\boldsymbol{h}}_{1}+\widetilde{\boldsymbol{h}}_{2}+\ldots+\widetilde{\boldsymbol{h}}_{K}\right\Vert }\boldsymbol{H}_{i}\nonumber\\
&\rightarrow\frac{\widetilde{\boldsymbol{h}}_{i}^{\mathrm{H}}\boldsymbol{H}_{i}}{\sqrt{K}}=\frac{\sqrt{N_{t}}}{\sqrt{K}}\boldsymbol{I}_{e},
\end{align}
where $\boldsymbol{I}_{e}=[\underbrace{0,\ldots,}_{e-1}1,\underbrace{\ldots, 0}_{N_{d}-e}]^{\mathrm H}$
if the $e$-th column of $\boldsymbol{U}_{r}$ is selected. Hence, we have
\begin{equation}
\left\Vert \boldsymbol{h}_{i}'\right\Vert =\left\Vert \frac{\sqrt{N_{t}}}{\sqrt{K}}\boldsymbol{I}_e\right\Vert \rightarrow\sqrt{\frac{N_{t}}{K}},\forall i.
\end{equation}
According to \textbf{Theorem }\ref{thm:MIMO}, the MSE and $r_{k}$
{can be obtained as}
\begin{align}
\mathrm{MSE} & =\left(\frac{1}{\sqrt{P}\sum_{i=1}^{K}\left\Vert \boldsymbol{h}_{i}'\right\Vert }\right)^{2}\sigma^{2}\rightarrow\frac{\sigma^{2}}{PKN_{t}},
\end{align}
\begin{align}
r_{k}=\frac{1}{l_{k}} & =c_{k}\sqrt{P}\sum_{i=1}^{K}\left\Vert \boldsymbol{h}_{i}'\right\Vert =\frac{\sum_{i=1}^{K}\left\Vert \boldsymbol{h}_{i}'\right\Vert }{K\left\Vert \boldsymbol{h}_{k}'\right\Vert }\rightarrow1.
\end{align}
Note that the achieved MSE is inversely proportional to $K$ and
$N_{t}$, irrespective of $N_{d}$ when $N_{t}\rightarrow\infty$ for
$N_{t}>N_{d}$. {This is because} the projections of the designed
$\boldsymbol{m}$ on sub-spaces $\mathrm{span}\left(\boldsymbol{H}_{i}\right),\forall i\in\mathcal{K}$
have the same power $\frac{\left\Vert \boldsymbol{m}\right\Vert }{K}=\frac{1}{K}$,
since $\boldsymbol{H}_{i}$ is a column full rank matrix with equal
singular values $N_{t}$ and the spanned sub-spaces are orthogonal.

\subsection{{Observations}}

{We observe the following facts from the above asymptotic analysis}:
\begin{itemize}
\item A larger number of antennas evidently provides more degree of freedom to {align the signals from the distributed devices. This in turn brings down the performance gain obtained by employing DLR, and the DLR ratio is pushed closer to $1$.}
\item When the number of antennas {increases} to infinity, the designed receive beamforming
$\boldsymbol{m}$ by simply summing up the normalized channel vectors
can achieve the lower bound $\mathrm{MSE}^{\mathrm{lbm}}$, as the channel
vectors under this case are asymptotically orthogonal with equal power.
\item The MSE is inversely proportional to {the} number of devices $K$ and {the} number of antennas $N_{t}$ in the SIMO case, while the MSE is inversely proportional
to $K^{2}N_{d}$ in the MISO case. The reason{, plainly,} is that the equivalent
channel power in the SIMO case is $\sqrt{\frac{N_{t}}{K}}$, whereas the
channel power in the MISO case is $\sqrt{N_{d}}$.
\item In the MIMO scenario, two cases, i.e. $N_{d}\rightarrow\infty$ for $N_{d}>N_{t}$,
and $N_{t}\rightarrow\infty$ for $N_{d}>N_{t}$, are considered. The {attained}
MSE is $\frac{\sigma^{2}}{PK^{2}N_{d}}$ in the former case, which
shares the same value with the MISO case regardless of $N_{t}$. Similarly,
the latter case is shown to have identical MSE, i.e., $\frac{\sigma^{2}}{PKN_{t}}$
with the SIMO scenario, independent of $N_{d}$. Such results
can be easily explained by comparing the equivalent channels in the MIMO
case and the channels in the MISO and SIMO scenarios.
\end{itemize}

\section{Simulation Results}\label{S:SR}

\begin{figure*}[!ht]
\centering
\subfigure[]{\includegraphics[width=0.499\columnwidth]{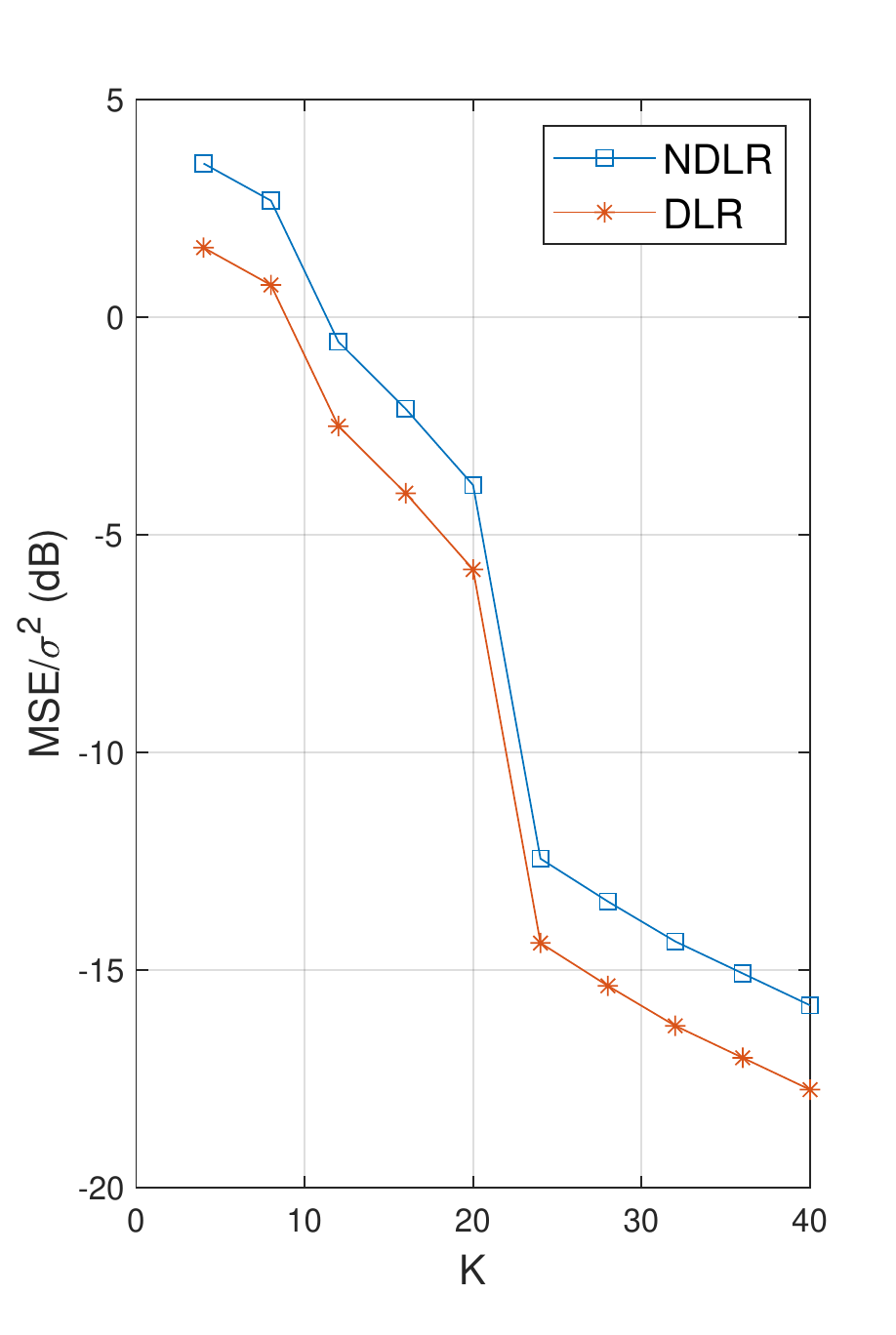}}
\subfigure[]{\includegraphics[width=0.499\columnwidth]{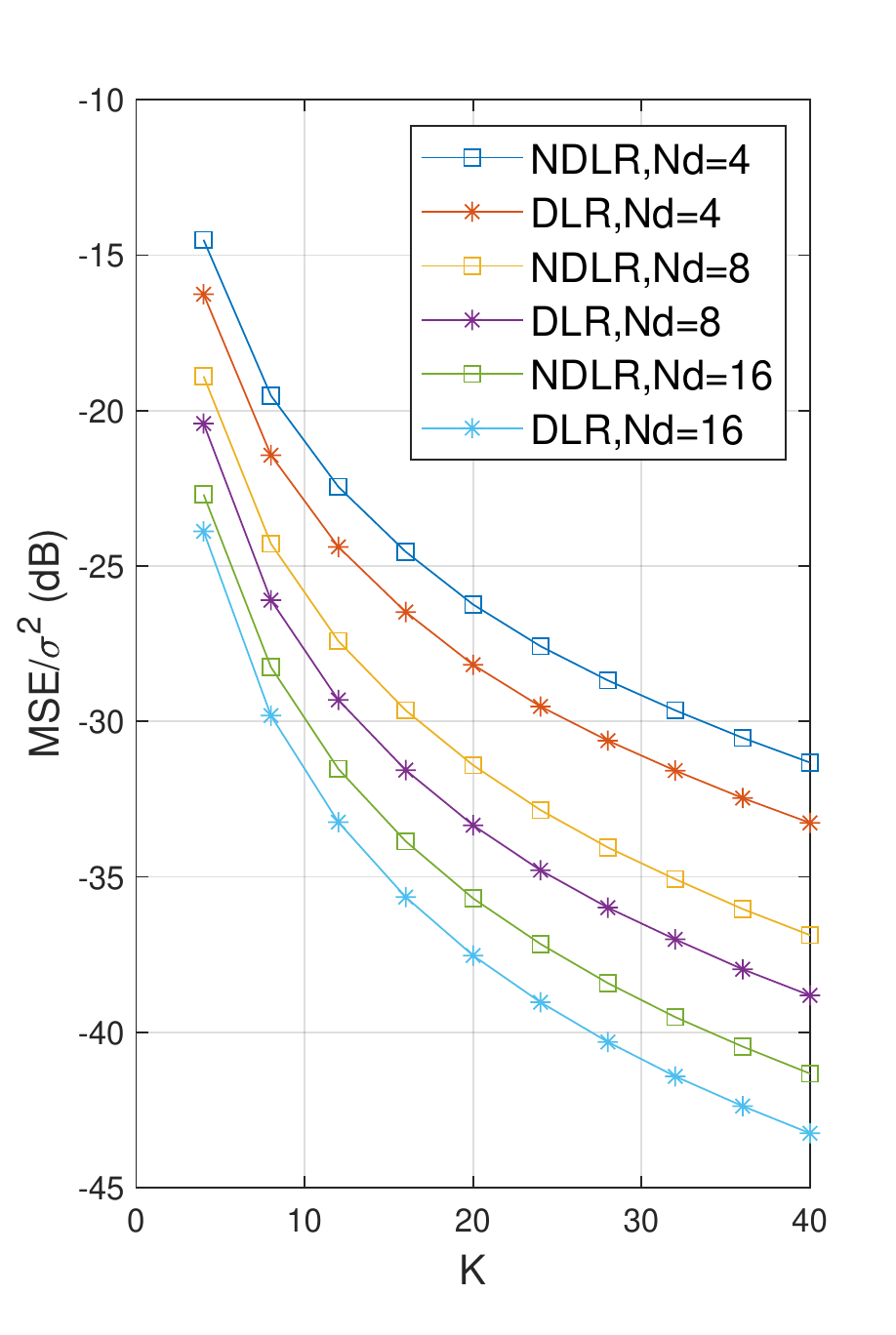}}
\subfigure[]{\includegraphics[width=0.499\columnwidth]{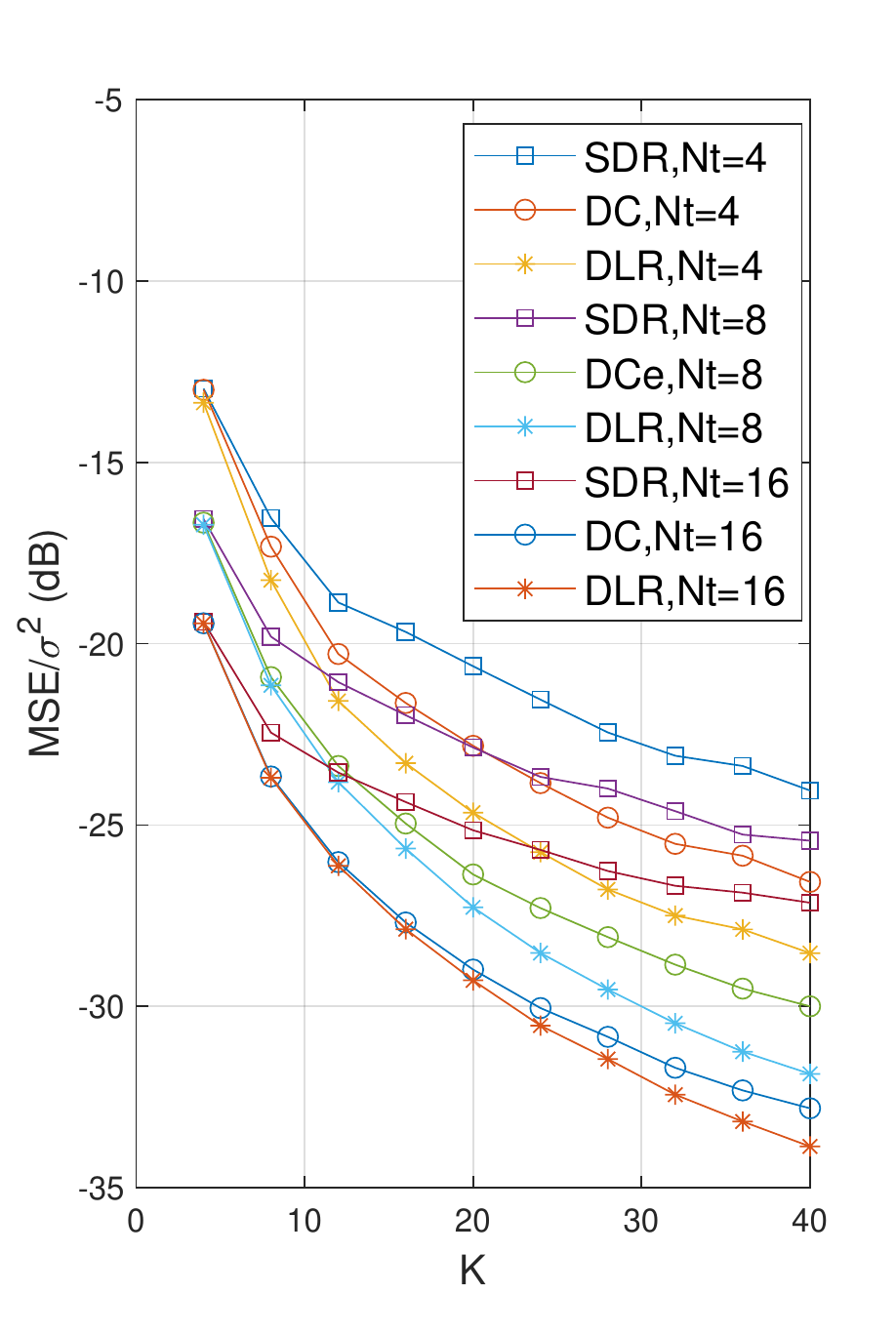}} \subfigure[]{\includegraphics[width=0.499\columnwidth]{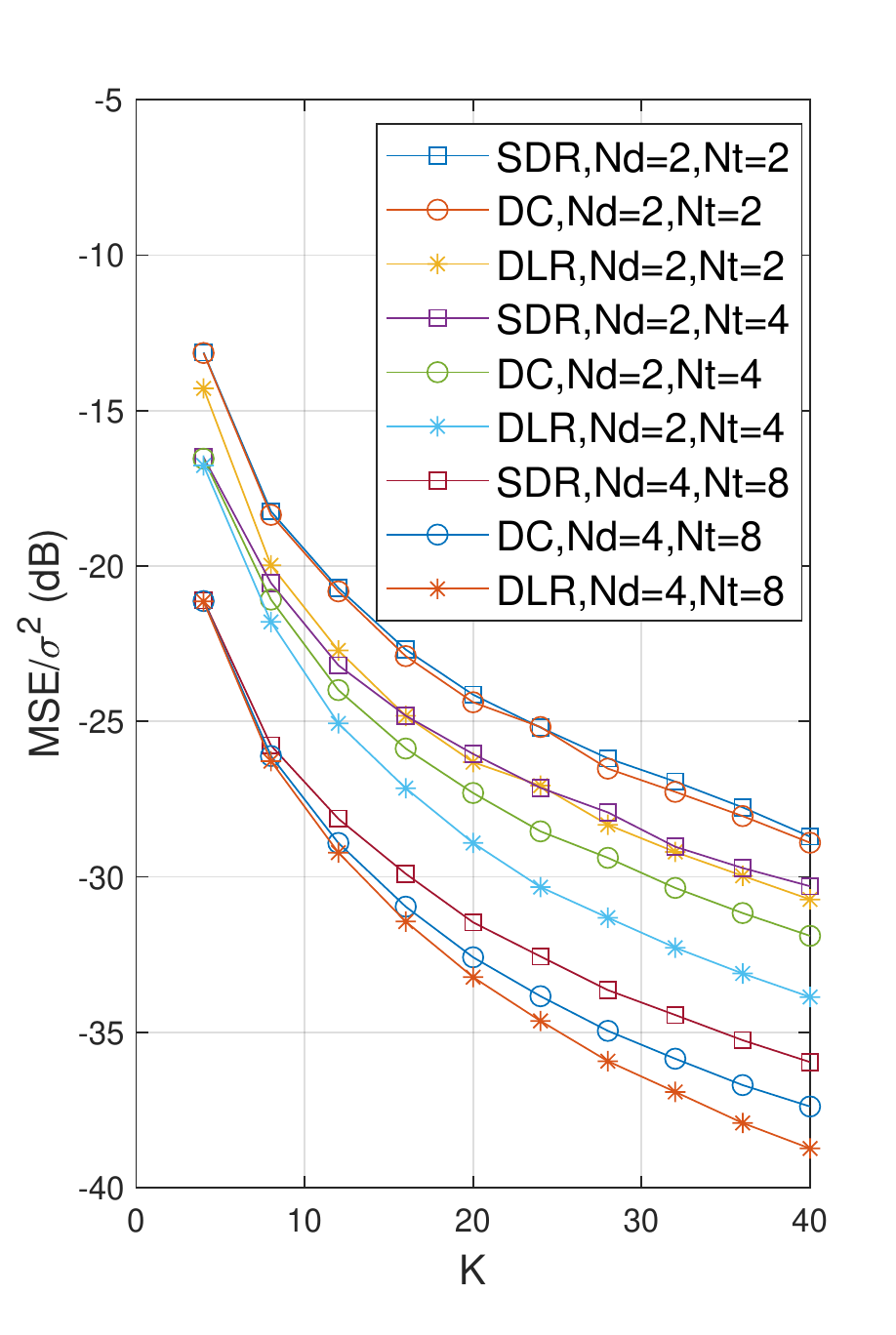}}
\caption{The aggregate error versus the number of devices $K$. (a) SISO scenario.
(b) MISO scenario with $N_{d}=4,8,16$. (c) SIMO scenario with $N_{t}=4,8,16$.
(d) MIMO scenario with $N_{d}=2,N_{t}=2$, $N_{d}=2,N_{t}=4$, $N_{d}=4,N_{t}=8$.}
\label{fig:MSE-V.S.-K}
\end{figure*}

Simulation results are given in this section to {demonstrate} the effectiveness
of the proposed DLR design and the performance of the proposed {near-optimal}
and closed-form receive beamforming solution when massive antennas are
applied. The proposed method is compared with {the existing approaches without
DRL}. The performance of the compared method is {labelled} as `NDLR' in the SISO and MISO scenarios, {which does not optimize the receive beamforming but only optimizes the transmit coefficients (vector) $ b_{k}\left(\boldsymbol{b}_{k}\right)$ using the method proposed in \cite{8364613}}.
In the SIMO and MIMO scenarios, the SDR and DC methods are compared to obtain receive beamforming vector $\boldsymbol{m}$, which are {labelled} as `SDR' and `DC'. It should be noted that all {the} devices participate in the update of the global model. We set the maximum
transmit power of each device $k$ as $P_{k}=0$ dB, which experiences
independent Rayleigh fading. To show the impact of DLR on the training and inference performance, we use FL to implement the classification tasks on MNIST and CIFAR10 datasets. Assume that the data stored at each device has equal {size. To ensure that the algorithm has adequate supplies of data to support feature extraction and meaningful learning, $K\geq20$ has to be satisfied.} MLP and ResNet18 neural networks are adopted to train on these two datasets for $200$ epochs,
where $\mu$ is set to $0.01$.

\subsection{Performance on MSE using DLR}

To {showcase} the effectiveness of the proposed DLR, we conduct simulations
under the SISO, MISO, SIMO and MIMO scenarios, {where the boundaries of DLR are set to $r_{\min}=1/1.2$ and $r_{\max}=1/0.8$.} Fig. \ref{fig:MSE-V.S.-K} displays $\mathrm{MSE}/\sigma^{2}$ with respect to the number of devices  $K$. It shows that the aggregate error decreases with
the increase of device number $K$ in four scenarios, which is due
to the averaging operation over $K$ devices. Specifically, more devices
indicate smaller scaling factor $\eta$ and therefore smaller error
from (\ref{eq:error}). The aggregate error is further reduced by
additionally considering DLR, compared to the methods utilizing only
on wireless resource {in \cite{8364613,kaiyang}}, which validates \textbf{Theorem \ref{thm:MISO}} and \textbf{Theorem \ref{thm:MIMO}}. It also reveals that performance gap extended with the increase of the number of devices $K$. The reason behind is that the increase of $K$ lead to a larger difference between the maximum and minimum signal power. Compared with the SISO scenario, multiple
antennas at the devices or/and the aggregator offer diversity
gain to combat  fading. Thereby, more antenna deployment
results in a smaller aggregate error while the performance gain obtained
by DLR is shrinking.

\begin{figure}[!htbp]
	\centering
	\includegraphics[width=0.99\columnwidth]{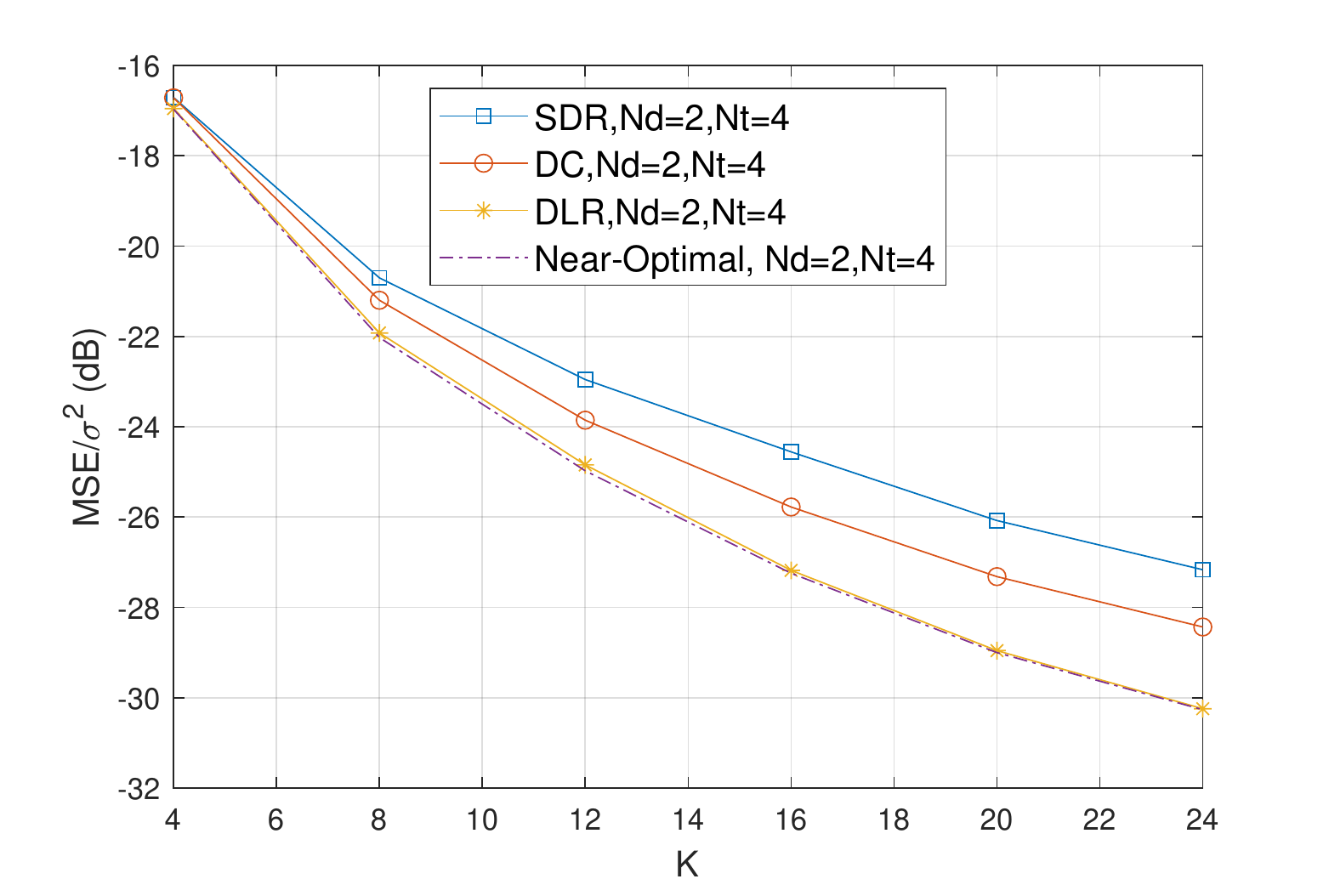}
	\caption{{The optimality of the proposed iterative method.}}
	\label{fig:near_opt}
\end{figure}

To showcase the optimality of the proposed iterative algorithm, we conduct numerical simulations to obtain the near-optimal solution in the MIMO scenario, as problem (\ref{eq:SIMO2}) is nonconvex and {the optimal solution is difficult to obtain. Note that} the near-optimal solution is obtained by initializing $10$ starting points and selecting the one with the minimum MSE. Fig.~\ref{fig:near_opt} {reveals} that the proposed iterative learning rate and receive beamforming algorithm can achieve relatively close performance with the near-optimal solution. In order to {examine} the impact of DLR, we {show} the channel gain/equivalent channel power and the corresponding DLR values in Fig.~\ref{fig:bar}. For {clear illustration}, the devices are indexed by the channel gain/equivalent channel gain in an ascending order.
As displayed in Fig. \ref{fig:bar}, the learning rate $\mu_{k}=r_{k}\mu$
is smaller for device $k$ with higher channel gain/equivalent channel
gain, while a larger learning rate is used for the local model update
with lower channel gain/equivalent channel gain, which can be explained
using (\ref{eq:ther}).

\begin{figure*}[!htbp]
\centering
\subfigure[]{\includegraphics[width=0.94\columnwidth]{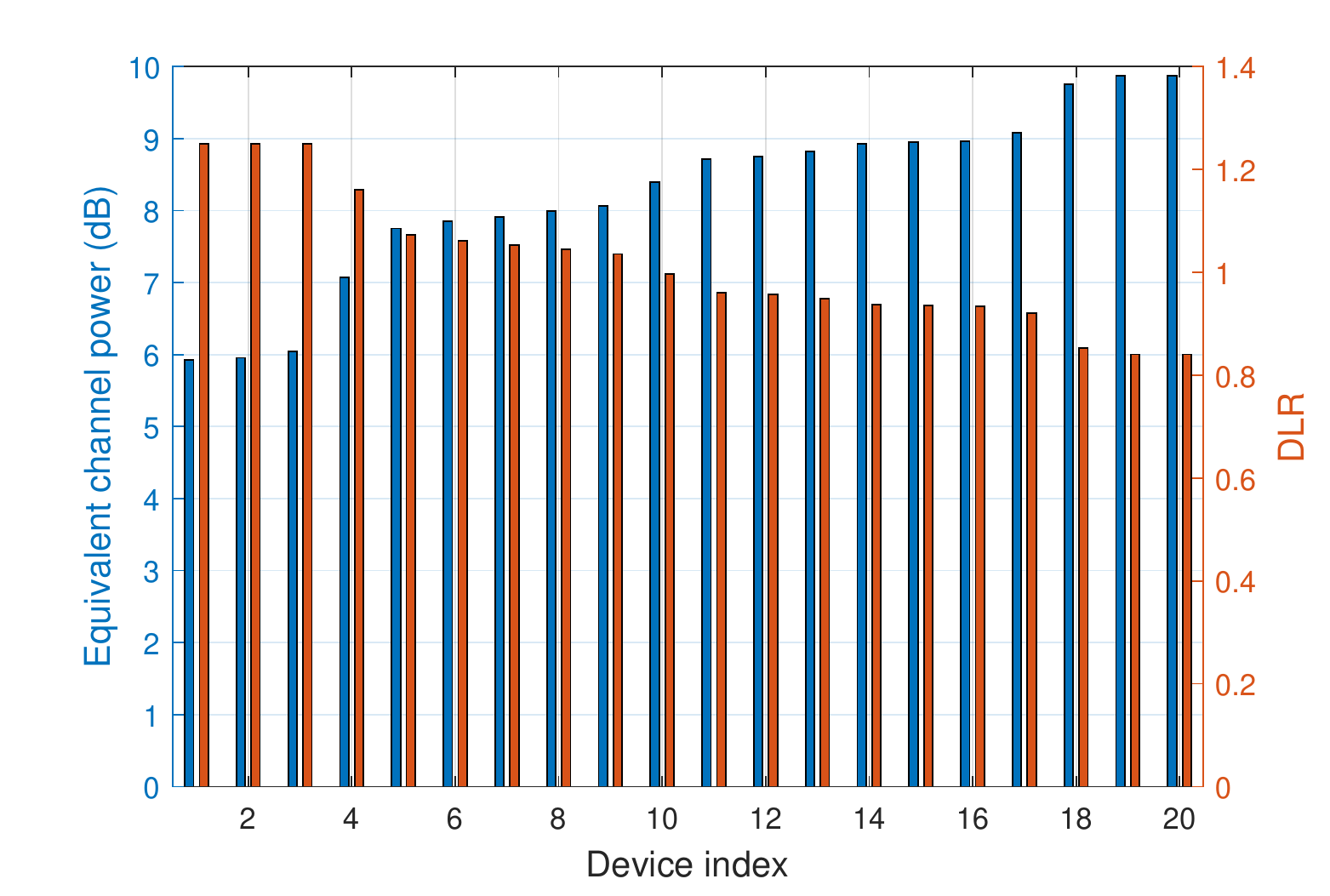}}
\subfigure[]{\includegraphics[width=0.94\columnwidth]{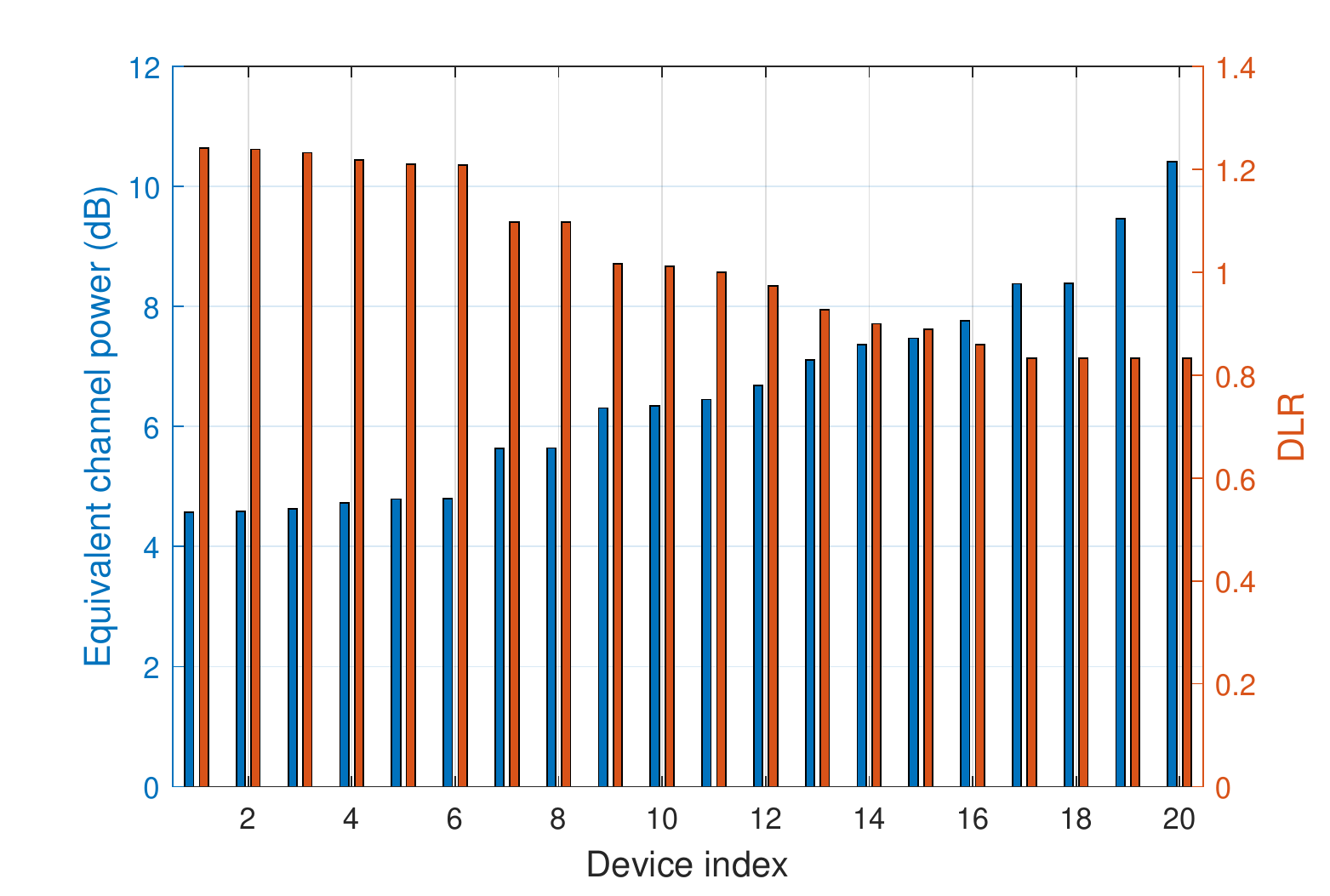}}
\caption{The equivalent channel power and the corresponding DLR. (a) MISO scenario
with $N_{d}=8$, (b) MIMO scenario with $N_{d}=4,N_{t}=4$.}
\label{fig:bar}
\end{figure*}

Simulations are conducted on different {boundaries} of DLR ratios,
i.e., $r_{\mathrm{min}}$ and $r_{\mathrm{max}}$, under four scenarios
to illustrate their impact. Note that the case of
$r_{\mathrm{min}}=r_{\mathrm{max}}=1$ is equivalent to conventional
methods without considering DLR. Fig.~\ref{fig:MSE_ratio} shows that
a larger range of DLR ratio leads to the decreasing trend of MSE.
In Fig. \ref{fig:MSE_ratio}(c) and Fig.~\ref{fig:MSE_ratio}(d),
the performance in the cases of $1/r_{\max}=0.4,1/r_{\min}=1.6$ and
$1/r_{\max}=0.6,1/r_{\min}=1.4$ is close when $K\leq10$. Recall
that the DLR is inversely proportional with the equivalent channel
gain. Such close performance can be explained by the reason that the
obtained receive beamforming vector $\boldsymbol{m}$ can well combat
the distortion due to the fading channels {when $K$ is small}.

\begin{figure*}[!htbp]
\centering
\subfigure[]{\includegraphics[width=0.499\columnwidth]{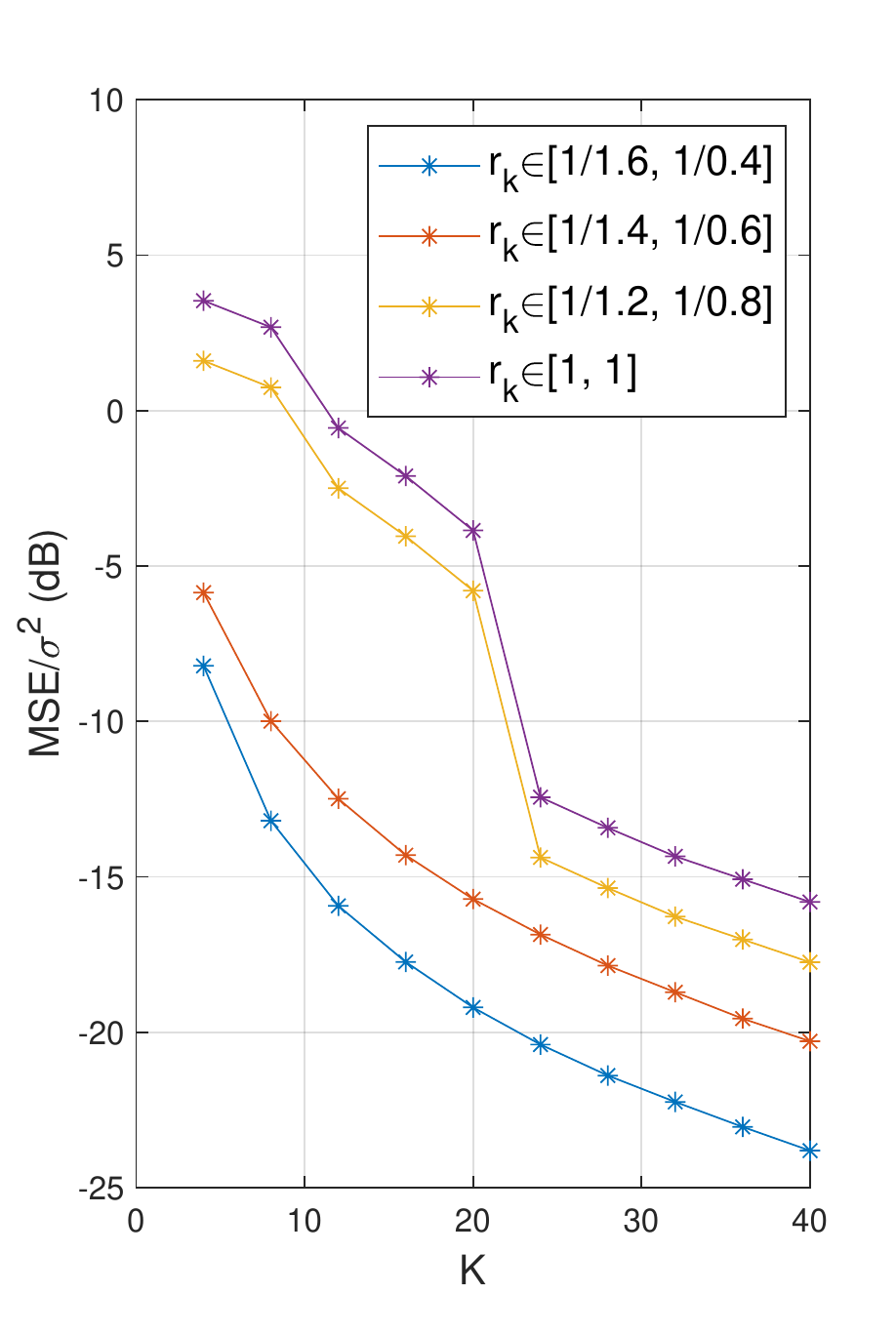}}
\subfigure[]{\includegraphics[width=0.499\columnwidth]{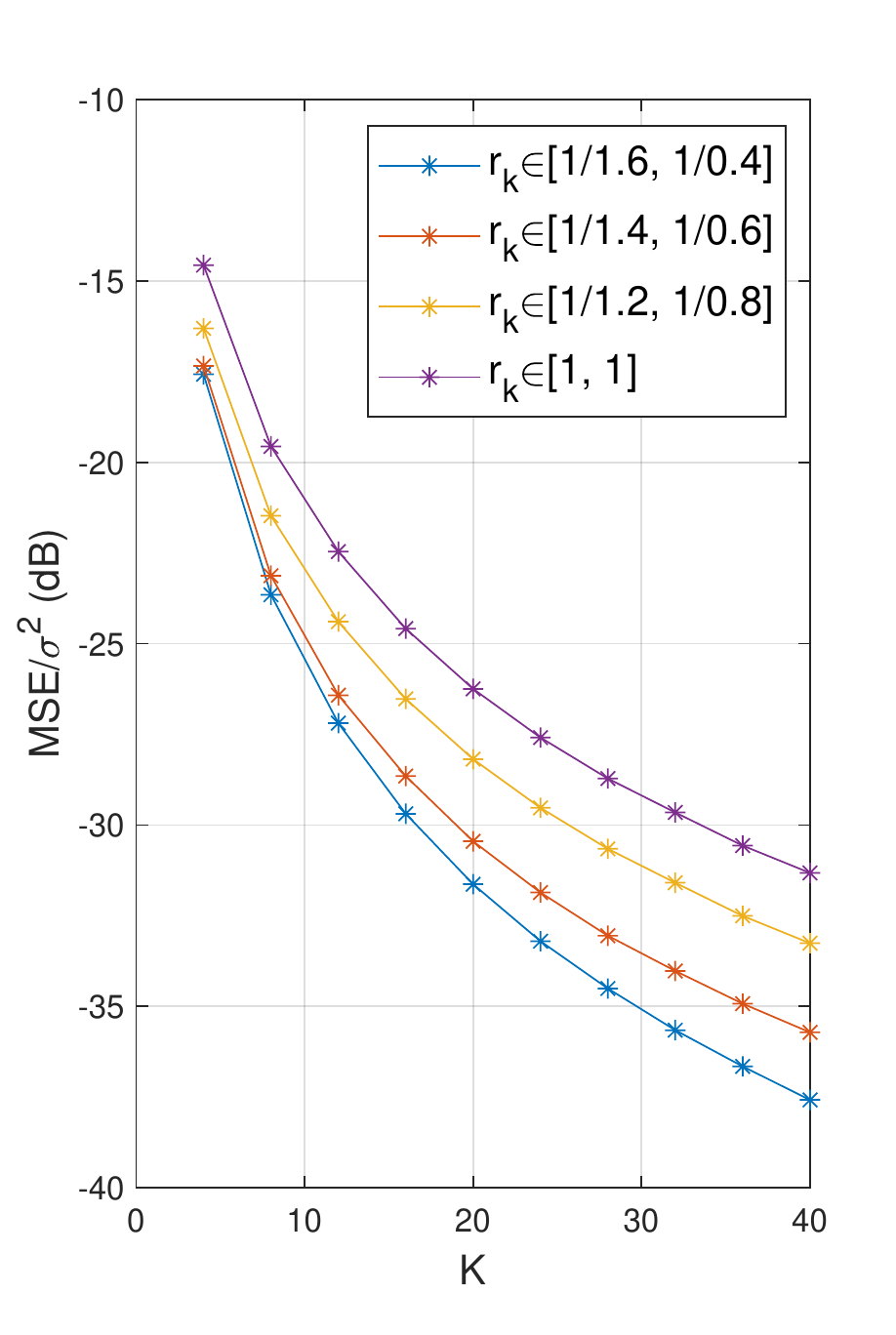}}
\subfigure[]{\includegraphics[width=0.499\columnwidth]{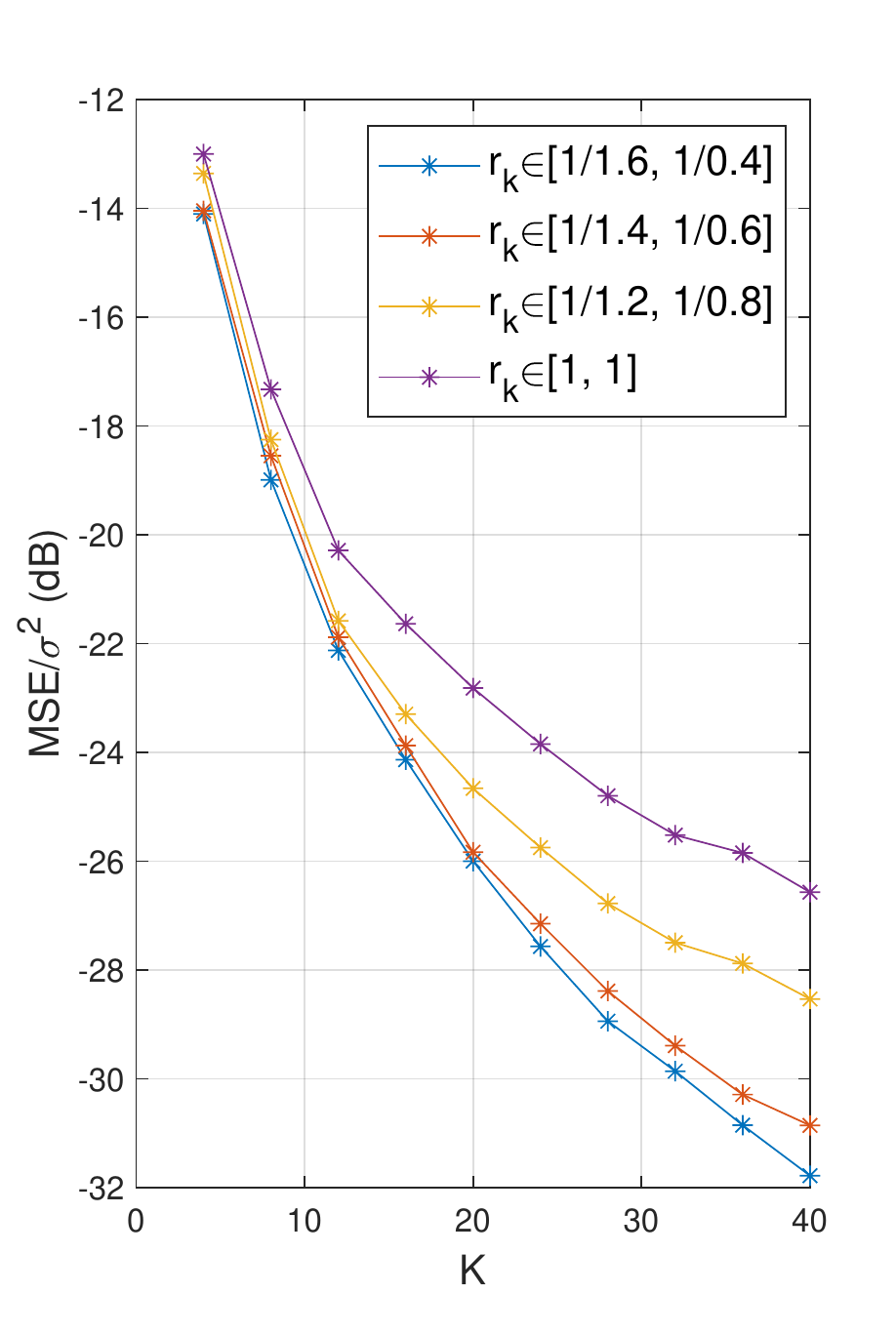}}
\subfigure[]{\includegraphics[width=0.499\columnwidth]{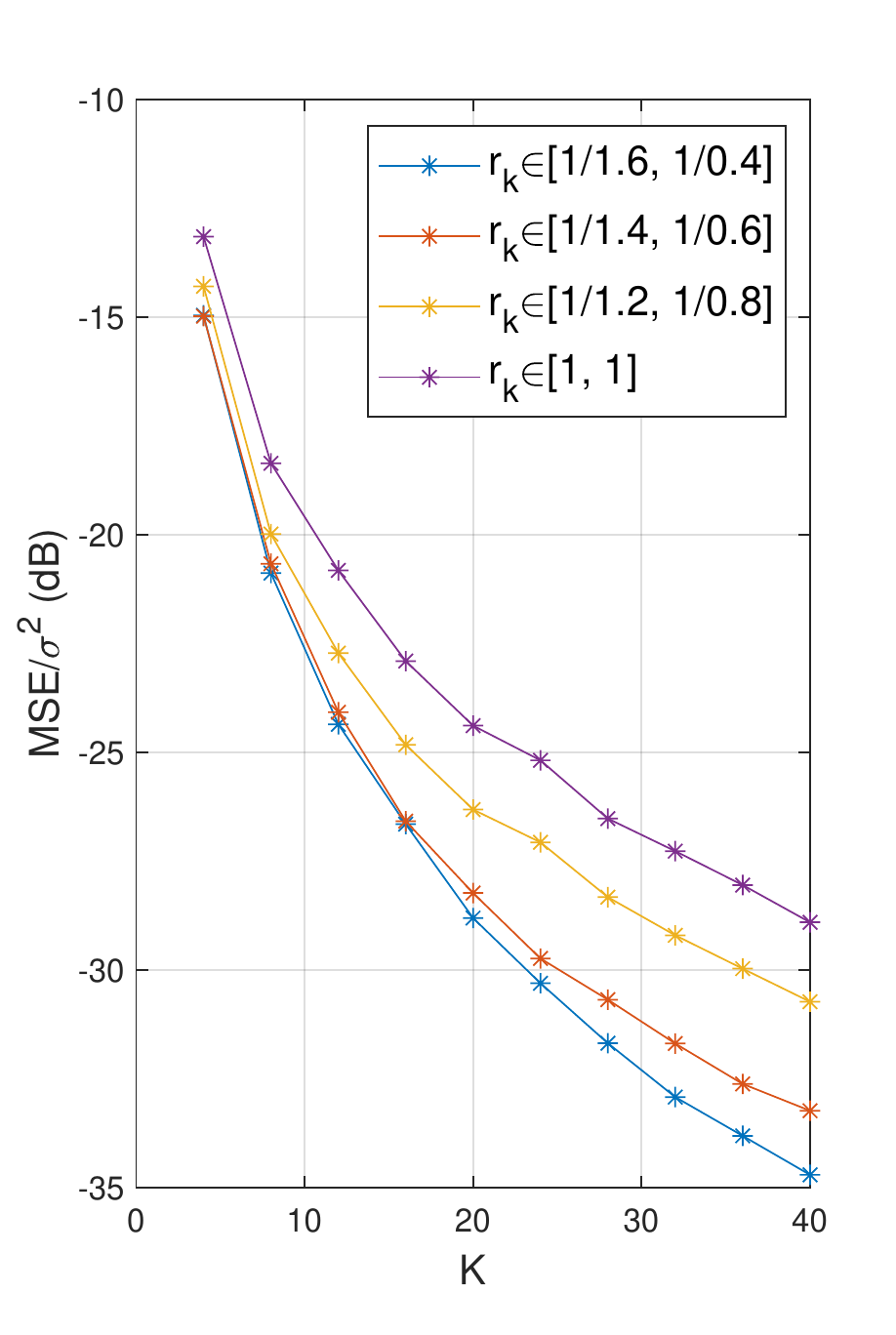}}
\caption{The impact of $r_{\max}$ and $r_{\min}$. (a) SISO scenario. (b)
MISO scenario with $N_{d}=4$. (c) SIMO scenario with $N_{t}=4$.
(d) MIMO scenario with $N_{d}=2,N_{t}=2$.}
\label{fig:MSE_ratio}
\end{figure*}

\subsection{Performance of Learning Task Using DLR}

To {investigate} the impact of DLR on the training and testing performance
of FL tasks, we utilize FL to perform classification tasks on MNIST and
CIFAR10 datasets. Fig.~\ref{fig:fl_loss_acc-1} gives the training
loss as well as the test accuracy on both datasets.
$K=20$ devices are involved in updating the global model, and the
parameter $a$ and the noise power $\sigma^{2}$ are respectively
set to $10$ dB and $0$ dB. Compared to CIFAR10 including $10$ classes
of color pictures, MNIST dataset comprising only black and white pictures
is {known to be} much more easier to learn. Thus, the accuracy trained on the MNIST can achieve $90\%$ very soon and {approach} almost $100\%$, while
the accuracy on CIFAR10 is lower with approximate $72\%$. Since the MSE with
DLR is smaller than the MSE with a fixed learning rate, its re-transmission
probability is smaller. {The DLR-based scheme} is shown to have slight higher test accuracy on both datasets compared {to} conventional methods using a fixed learning rate, i.e., $r_{\mathrm{min}}=r_{\mathrm{max}}=1$. This {is due to the fact} that a increased learning rate owing to adaption to the fading channels can help escape the saddle point which is known as the difficulty in minimizing the loss. In addition, a larger range of DLR may result in a bigger variance of the training loss and the test accuracy, which implies that proper boundaries of DLR ratio should be chosen.

{Further numerical} simulations are conducted under the MISO and MIMO scenarios with
$K=4, 12, 20$ devices, {where $N_d=4$, and $N_d=2, N_t=4$, respectively}. The DLR ratio {boundaries are set to $r_{\min}=1/1.2$ and  $r_{\max}=1/0.8$ }, and the noise power
is set to $10$ dB. The reported accuracy performance on both datasets
is given in Table \ref{tab:Rep_accuracy}. {The assumption of sufficient data suggests that the total size of data under $K=4$ and $K=12$ cases is insufficient.} As a result, the test accuracy on both datasets is lower when $4$ and $12$ devices participate in aggregating the global model, {compared with $20$ devices. It is worth noting that the reported accuracy using DLR may be smaller than that with a fixed learning rate due to the variance.}
Therefore, the simulation results in Fig. \ref{fig:fl_loss_acc-1}
and Table \ref{tab:Rep_accuracy} demonstrate that the proposed DLR can slightly improve the learning and inference performance compared with the {fixed-learning-rate-based approach}.

\begin{figure*}[!htpb]
\centering
\subfigure[]{\includegraphics[width=0.94\columnwidth]{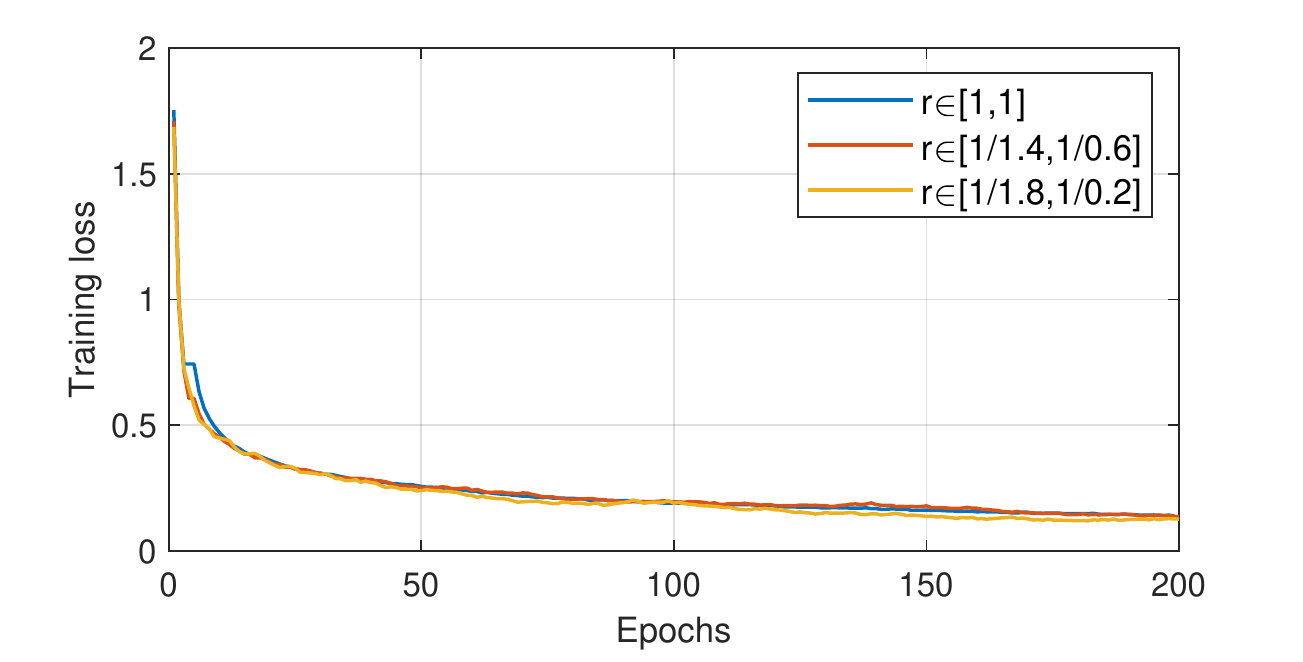}}
\subfigure[]{\includegraphics[width=0.94\columnwidth]{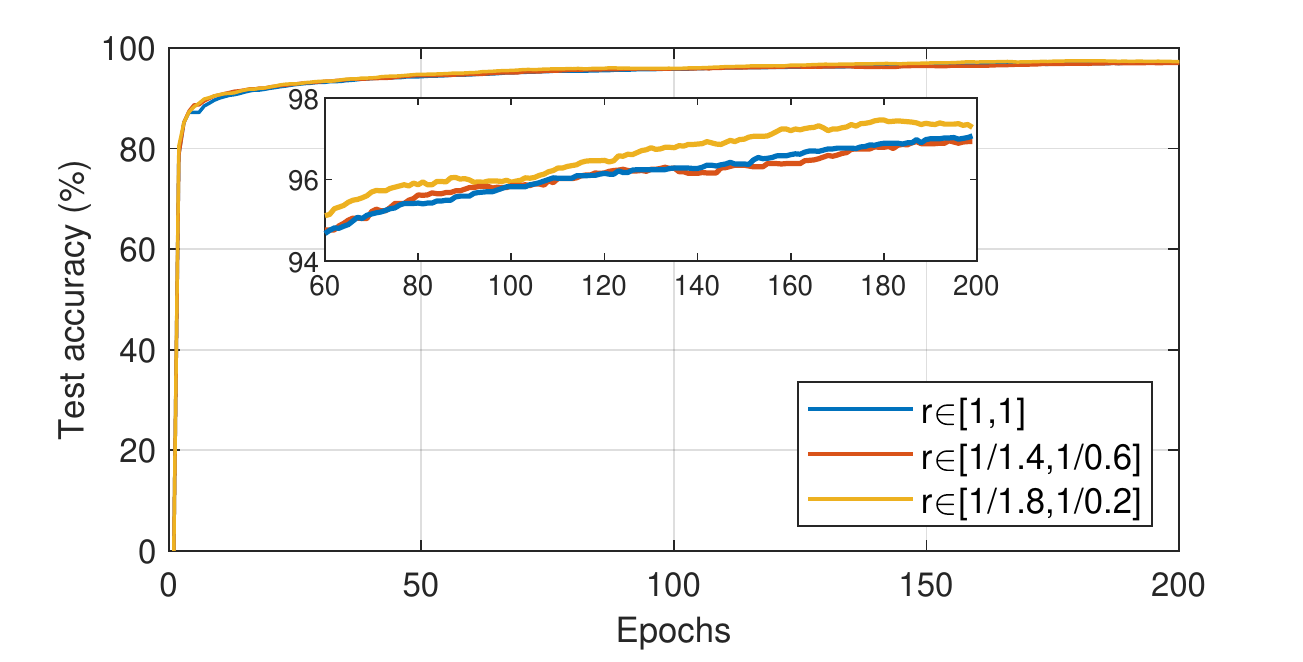}}
\subfigure[]{\includegraphics[width=0.94\columnwidth]{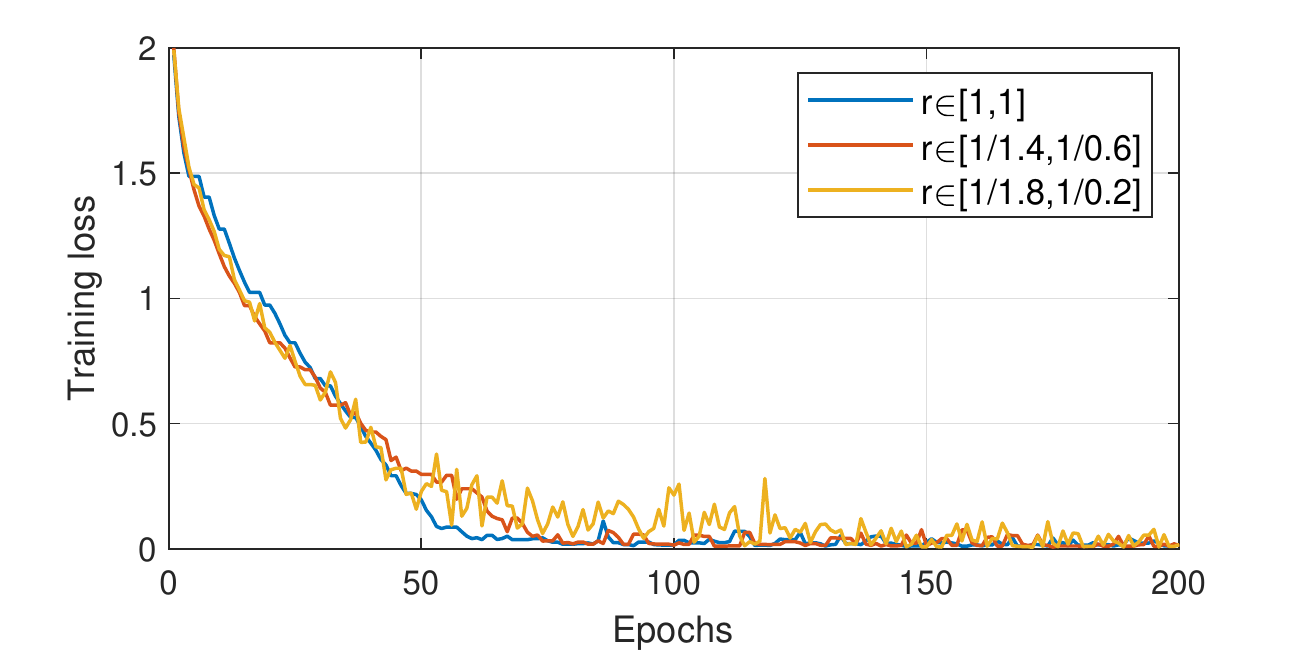}}
\subfigure[]{\includegraphics[width=0.94\columnwidth]{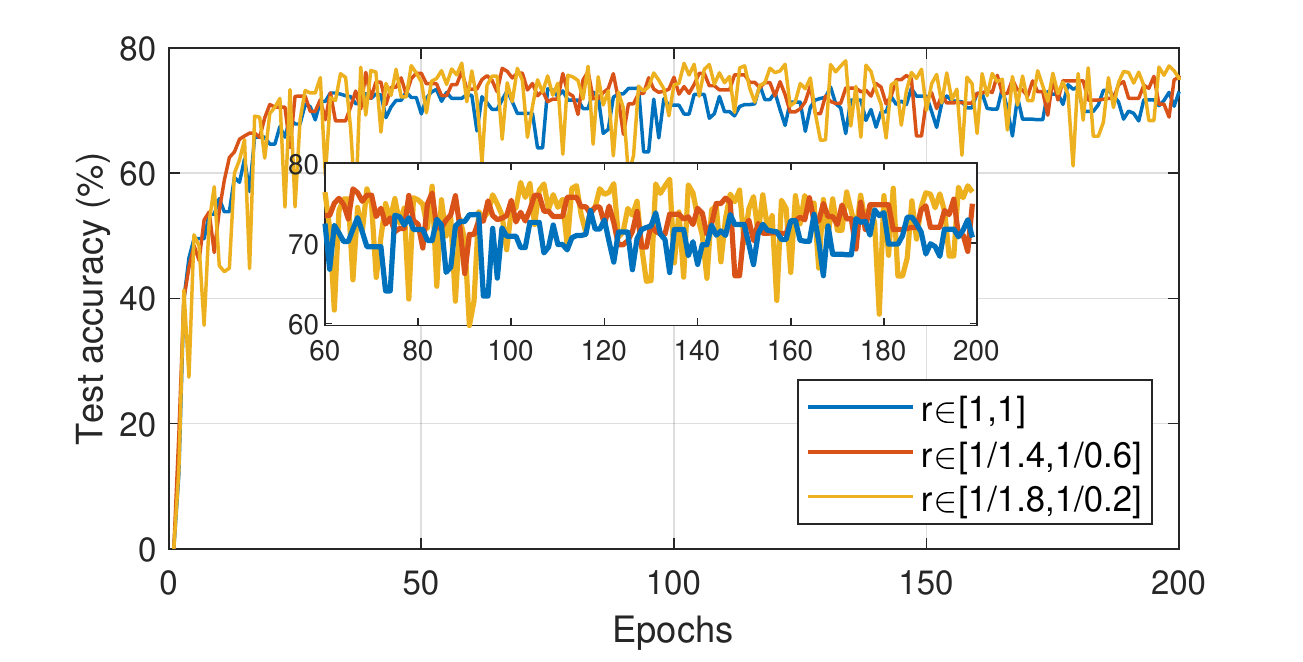}}
\caption{The training and testing performance with different $r_{\mathrm{max}}$
and $r_{\mathrm{min}}$. (a) Training loss on MNIST dataset. (b) Accuracy
performance on MNIST dataset. (c) Training loss on CIFAR10 dataset.
(d) Accuracy performance on CIFAR10 dataset.}
\label{fig:fl_loss_acc-1}
\end{figure*}

\begin{table*}[t]
\caption{Reported accuracy on MNIST and CIFAR10 with $200$ epochs under MISO
and MIMO scenarios.}
\centering%
\begin{tabular}{>{\centering}m{2.5cm}|>{\centering}m{3cm}|>{\centering}m{1cm}|>{\centering}m{1cm}|>{\centering}m{1cm}|>{\centering}m{1cm}|>{\centering}m{1cm}|>{\centering}m{1cm}}
\hline
\multirow{2}{2.5cm}{Scenario} & Number of devices & \multicolumn{2}{c|}{$K=4$} & \multicolumn{2}{c|}{$K=12$} & \multicolumn{2}{c}{$K=20$}\tabularnewline
\cline{2-2} \cline{3-3} \cline{4-4} \cline{5-5} \cline{6-6} \cline{7-7} \cline{8-8}
 & Dataset & MNIST & CIFAR10 & MNIST & CIFAR10 & MNIST & CIFAR10\tabularnewline
\hline
\multirow{2}{2.5cm}{MISO\\
$N_{d}=4$\\
} & Fixed learning rate & 91.12\% & 53.11\% & 96.26\% & 62.68\% & 97.01\% & 68.72\%\tabularnewline
\cline{2-8} \cline{3-8} \cline{4-8} \cline{5-8} \cline{6-8} \cline{7-8} \cline{8-8}
 & DLR & 93.29\% & 50.27\% & 96.35\% & 64.68\% & 97.17\% & 72.12\%\tabularnewline
\hline
\multirow{2}{2.5cm}{MIMO \\
$N_{d}=2$, $N_{t}=4$\\
} & Fixed learning rate & 93.11\% & 50.96\% & 96.71\% & 63.39\% & 96.98\% & 71.29\%\tabularnewline
\cline{2-8} \cline{3-8} \cline{4-8} \cline{5-8} \cline{6-8} \cline{7-8} \cline{8-8}
 & DLR & 93.90\% & 54.86\% & 96.66\% & 64.94\% & 96.94\% & 73.86\%\tabularnewline
\hline
\end{tabular}
\label{tab:Rep_accuracy}
\end{table*}

\subsection{Performance of the Proposed Closed-form Receive Beamforming Solution}

{Now we move on to discuss the impact of {the number of antennas} on the MSE performance and verify the asymptotic analysis and the proposed closed-form receive beamforming design. Different numbers of antennas at the devices and the {aggregator} under the MISO, SIMO and MIMO scenarios are considered. Due to the lack of space, we restrict ourselves to the case of $K=2$ and $K=4$ only.} The line
labeled as `Analysis' is the derived theoretical MSE, and `Proposed' is the performance using the proposed simple closed-form receive beamforming design.

As shown in Fig.~\ref{fig:MSE-V.S.-ant}, the increase of {the number of antennas}  leads to a {reduction of} aggregate error since higher beamforming gain is {achieved}. The performance gap shrinks between the proposed DLR and fixed learning rate methods {for $r_k\rightarrow 1$ when the number of antennas go to infinity. More specifically,} in the MISO scenario, more antennas
equipped at the devices lead to smaller differences on the channel
gain among devices. Thus the performance gain obtained by DLR is smaller
with respect to fixing learning rate. Both `DLR' and `NDLR' approach
the `Analysis' performance eventually, which verifies the analysis
under the MISO case. In the SIMO and MIMO scenarios, more antennas
at the aggregator provide more freedom to align the received signals,
which leads to the reduced performance gap between `DLR' and `DC'.

{Fig. \ref{fig:MSE-V.S.-ant} also shows that the MSE obtained by the proposed closed-form receive beamforming design is approaching the theoretical {bound} when massive antenna are applied. Besides, the MSE without consideration of DLR, i.e., `NDLR' in the MISO case and `DC' in both SIMO and MIMO cases, is getting closer {to} `Analysis', which implies that the equality of (a) and (b) in both (\ref{eq:MISO_ther}) and (\ref{eq:MIMO_ther}) are guaranteed.} Based
on the theoretical analysis in {Section~\ref{S:AAna}}, when the {the number of receive antennas (or transmit antennas)} goes to infinity, the analyzed MSE is the same
for both MIMO and MISO (or SIMO). {From Fig.~\ref{fig:MSE-V.S.-ant} we observe} that the number of transmit (or receive) {antennas} can affect
the changing {rate} of MSE curve when {the number of the} receive (or transmit) antennas
goes to infinity. Therefore, {the asymptotic analysis is verified and the effectiveness of the proposed closed-form receive beamforming design in Section~\ref{S:AAna} is confirmed.}

\begin{figure*}[!htpb]
\centering
\subfigure[]{\includegraphics[width=0.499\columnwidth]{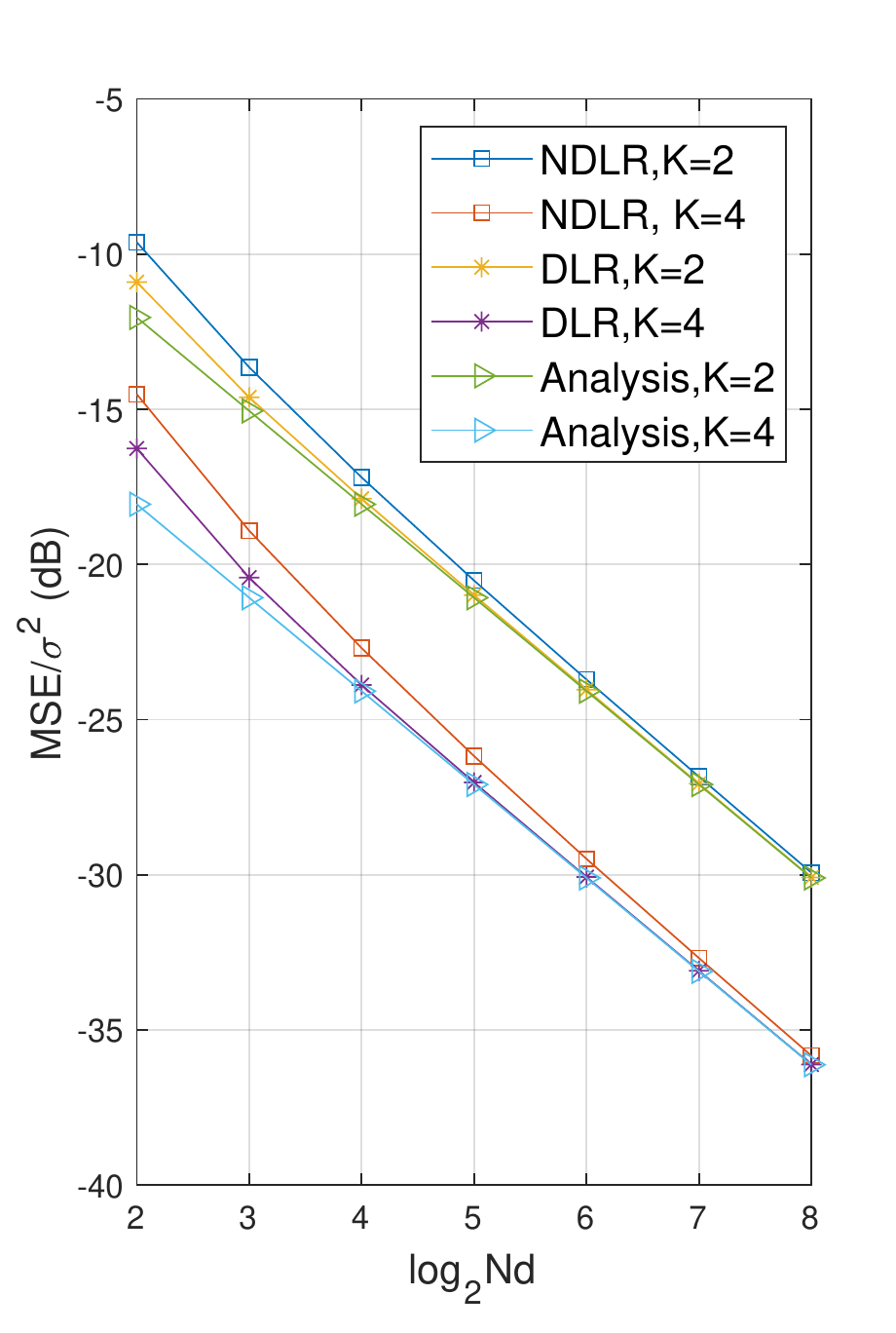}}
\subfigure[]{\includegraphics[width=0.499\columnwidth]{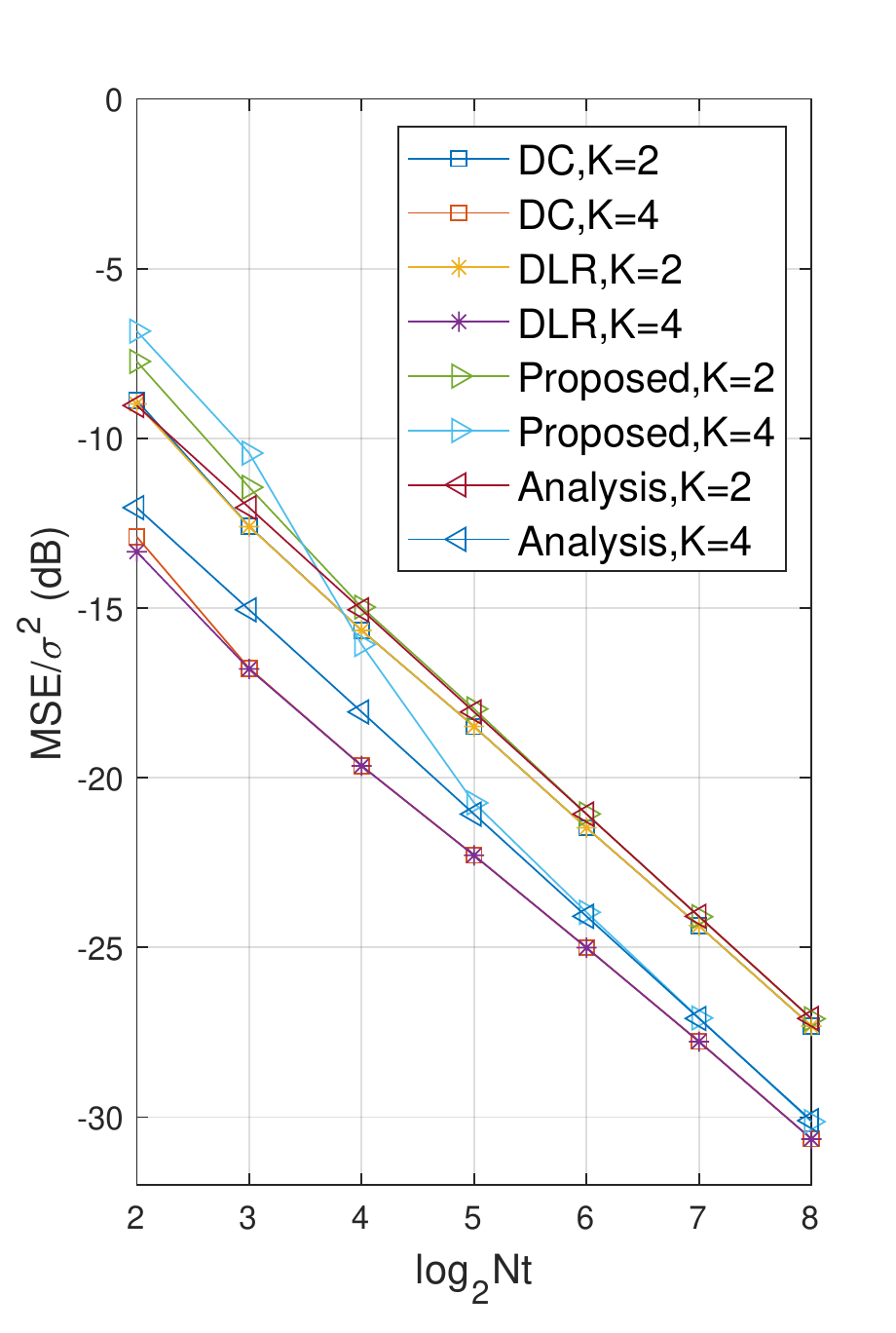}}
\subfigure[]{\includegraphics[width=0.499\columnwidth]{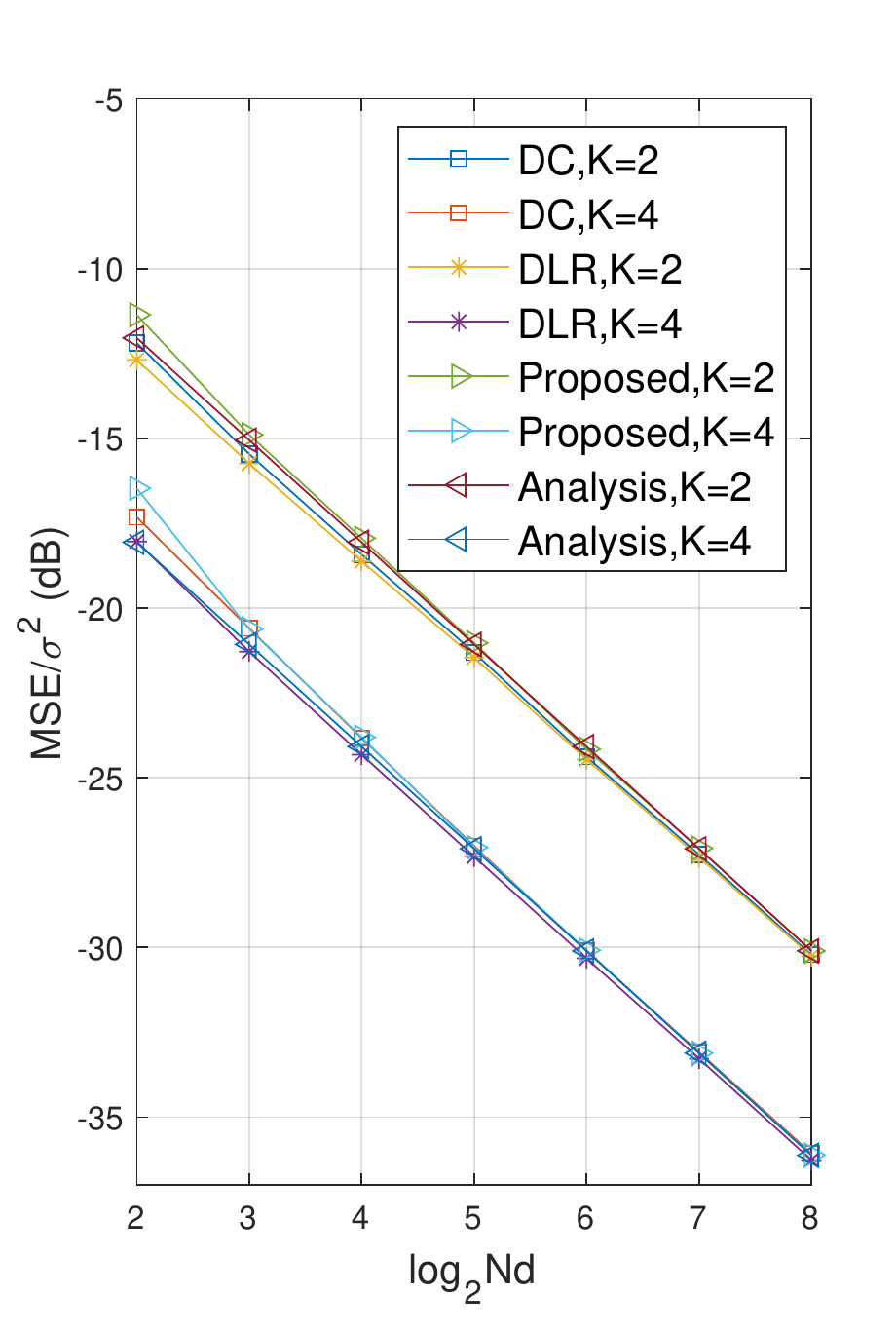}}
\subfigure[]{\includegraphics[width=0.499\columnwidth]{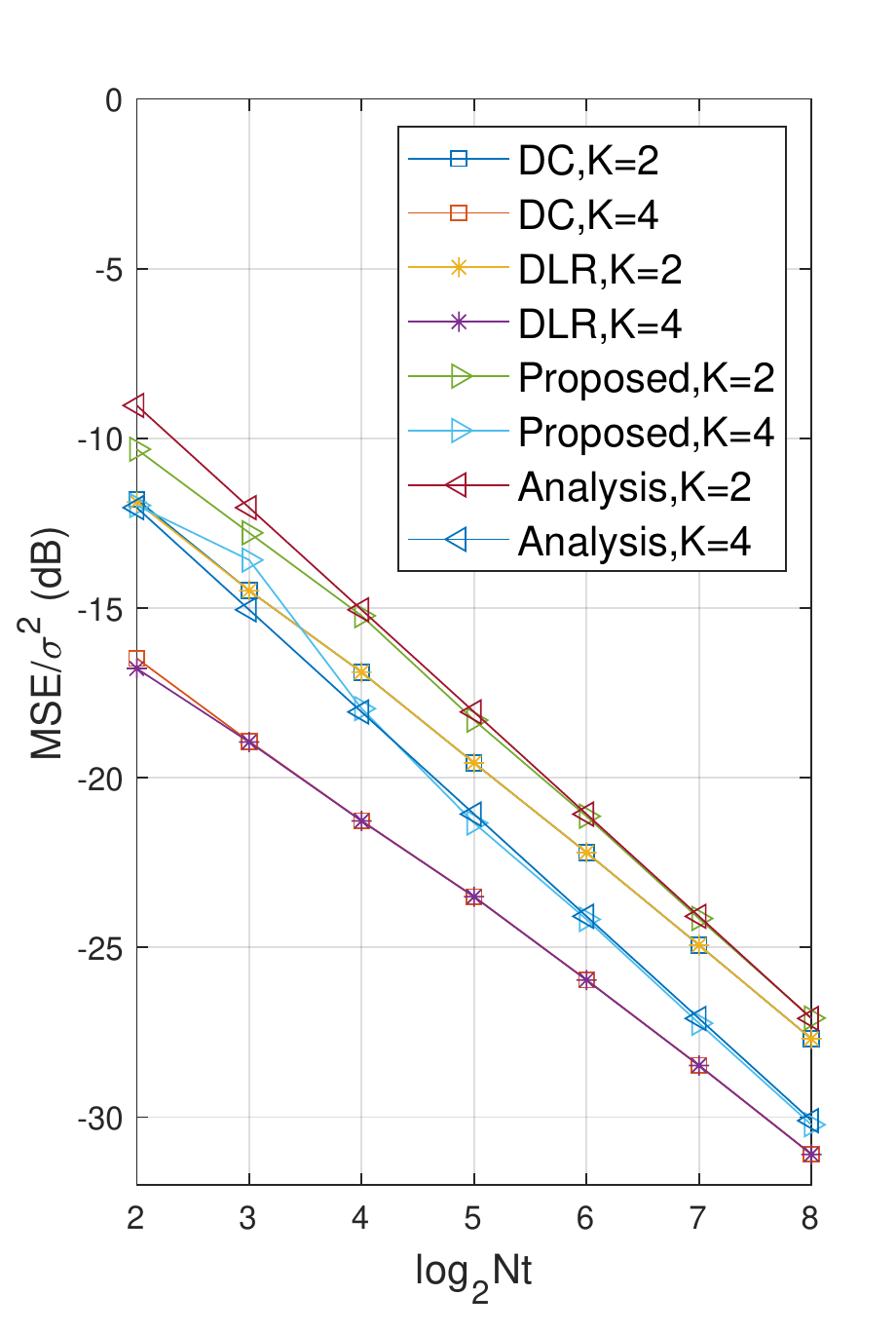}}
\caption{The impact of antenna number and the performance of the proposed closed-form
receive beamforming design. (a) MISO scenario. (b) SIMO scenario.
(c) MIMO scenario with $N_{t}=2$. (d) MIMO scenario with $N_{d}=2$.}
\label{fig:MSE-V.S.-ant}
\end{figure*}

\section{Conclusions}\label{S:conclu}

{In this paper, we {incorporated} the Aircomp technique in implementing distributed learning tasks to significantly improve the communication efficiency. Minimizing the aggregate error due to the fading and noisy channels is critical as a large error can lead to poor training and inference performance. Different from existing works that mainly {optimized} the wireless resources to align the received signals, we first {proposed} to utilize optimizable learning rates and proposed DLR to adapt to the fading channels. The problem {was} formulated under the MISO and MIMO scenarios. A closed-form solution and an iterative method {were} proposed respectively for each case. The simulation results have validated the effectiveness of the proposed DLR in terms of both the MSE performance and the testing accuracy on the MNIST and CIFAR10 datasets. Asymptotic analyses in the MISO, SIMO and MIMO scenarios {were} also provided to address the problem of massive antenna deployment as well as to derive the theoretical bounds. On this basis, a {near-optimal} and closed-form receive beamforming design {was} proposed by simply summing up the channel vectors. The feasibility and effectiveness of the proposal are verified by extensive numerical simulations.}

\bibliographystyle{IEEEtran}
\bibliography{mybib}

%\begin{IEEEbiography}{}
%\end{IEEEbiography}

\end{document}